\documentclass[prd,aps,amsfonts,eqsecnum,nofootinbib,notitlepage,longbibliography]{revtex4-1}

\usepackage{graphicx}
\usepackage{tabularx}
\usepackage{braket}
\usepackage{amssymb}
\usepackage{amsmath}
\usepackage{bm}
\usepackage[driverfallback=dvipdfm]{hyperref}
\hypersetup{colorlinks=true,breaklinks,urlcolor=blue,linkcolor=blue,citecolor=blue}

\newcommand{\red}[1]{\textcolor{black}{#1}}

\newcommand{\comment}[1]{}

\newcommand{\lr}[1]{\left( #1\right)}
\newcommand{\mlr}[1]{\left[ #1\right]}

\newcommand{\norm}[1]{\left\lVert#1\right\rVert}

\newcommand{\ee}{\mathrm{e}}

\makeatletter
\renewcommand{\p@subsection}{}
\renewcommand{\p@subsubsection}{}
\makeatother

\usepackage{amsthm}
\newtheorem{thm}{Theorem}
\numberwithin{thm}{section}
\newtheorem{cor}[thm]{Corollary}
\newtheorem{lem}[thm]{Lemma}

\newtheorem{prop}[thm]{Proposition}

\begin{document} 

\title{Finite speed of quantum information in models of interacting bosons at finite density}

\author{Chao Yin}
\email{chao.yin@colorado.edu}
\affiliation{Department of Physics and Center for Theory of Quantum Matter, University of Colorado, Boulder CO 80309, USA}

\author{Andrew Lucas}
\email{andrew.j.lucas@colorado.edu}
\affiliation{Department of Physics and Center for Theory of Quantum Matter, University of Colorado, Boulder CO 80309, USA}

\date{May 17, 2022}

\begin{abstract}
We prove that quantum information propagates with a finite velocity in any model of interacting bosons whose (possibly time-dependent) Hamiltonian contains spatially local single-boson hopping terms along with arbitrary local density-dependent interactions.  More precisely, with density matrix $\rho \propto \exp[-\mu N]$ (with $N$ the total boson number), ensemble averaged correlators of the form $\langle [A_0,B_r(t)]\rangle $, along with out-of-time-ordered correlators, must vanish as the distance $r$ between two local operators grows, unless $t \ge r/v$ for some finite speed $v$.  In one dimensional models, we give a useful extension of this result that demonstrates the smallness of all matrix elements of the commutator $[A_0,B_r(t)]$ between finite density states if $t/r$ is sufficiently small.    Our bounds are relevant for physically realistic initial conditions in experimentally realized models of interacting bosons. In particular, we prove that $v$ can scale no faster than linear in number density in the Bose-Hubbard model: this scaling matches previous results in the high density limit.  The quantum walk formalism underlying our proof provides an alternative method for bounding quantum dynamics in models with unbounded operators and infinite-dimensional Hilbert spaces, where Lieb-Robinson bounds have been notoriously challenging to prove.
\end{abstract}

\maketitle

\tableofcontents

\section{Introduction}
In Einstein's theory of relativity, information cannot travel faster than the speed of light $c$.  However, there can also be emergent speed limits (such as a speed of sound which controls auditory signaling) which are much slower than $c$.  In quantum mechanical systems, it was first proved by Lieb and Robinson \cite{Lieb1972} that there is a finite speed of quantum information in local lattice models with finite-dimensional Hilbert spaces (on any given site).   Analogously to the relativistic setting, it is said that these local lattice models have a ``Lieb-Robinson light cone" -- information propagates with a finite velocity $v$, and signals cannot be sent between ``spacelike separated" qubits, separated by a distance $x>vt$. Especially in recent years, many authors have qualitatively improved upon the original bounds of Lieb and Robinson, both in local lattice models \cite{Hastings:2005pr,Nachtergaele_2006,Bentsen_2019,chen2019operator,PRXQuantum.1.010303}, in dissipative and non-unitary dynamics \cite{Poulin_2010}, models with power-law interactions \cite{fossfeig,else,Tran_2019_polyLC,chen2019finite,Kuwahara:2019rlw,Tran:2020xpc,Kuwahara_OTOC,chen2021concentration,tran2021optimal,chen2021optimal}, in all-to-all interacting models \cite{guo2019signaling,Yin:2020pjd}, in semiclassical spin models \cite{PRXQuantum.1.010303,Yin:2020oze}, and even in microscopic toy models of quantum gravity \cite{Lucas:2019cxr,Lucas:2020pgj}.

However, it has proven notoriously difficult to find rigorous bounds on quantum dynamics in models with infinite dimensional Hilbert spaces.  This is not a simple mathematical curiosity, avoidable in any practical physical setting: any quantum mechanical system with conventional bosonic degrees of freedom, such as photons or phonons, has an infinite dimensional Hilbert space arising from the bosonic degrees of freedom.  Indeed, a simple model demonstrates that quantum information can propagate arbitrarily fast in certain bosonic systems \cite{gross}, and so any bound on dynamics must be restricted to special kinds of bosonic models.  Nevertheless, the model of \cite{gross} is somewhat unusual -- the ``hopping terms" in the Hamiltonian can annihilate or create two bosons, rather than moving a single boson from one site to another.  Could it be the case that in more physically relevant bosonic models, there \emph{is} a finite speed of information?   

While initial progress towards answering this question (ideally in the affirmative) was restricted to the analysis of systems with interacting bosons with bounded interactions \cite{Nachtergaele_2008}, or to classical models \cite{Raz_2009}, more recent work has been able to bound special classes of commutators in interacting models which have boson-spin interactions \cite{LRion} of a very special kind, relevant to cavity quantum electrodynamics \cite{Leroux_2010} or trapped ion crystals \cite{Britton_2012}.   Attempts to derive a finite velocity on information propagation have also been successfully made when restricting to states with a finite number of \emph{total} bosons \cite{schuch} (yet vanishing boson density in the thermodynamic limit).  In macroscopic quantum states with sufficiently low number density of bosons, a recently derived bound shows that the shortest time $t$ in which information can propagate a distance $r$ is $t\sim r/\log^2 r$ \cite{kuwahara2021liebrobinson}, in models where the interactions are density-dependent.   The result of \cite{kuwahara2021liebrobinson}, which is relevant to most physically realized models of interacting boson models, roughly suggests that the velocity of information grows with time as $v \lesssim \log^2 t$.  This is almost -- though not quite -- a ``linear light cone" in the same spirit as the Lieb-Robinson bounds on local spin chains.

Despite the very longstanding theoretical challenge in establishing the finiteness of the speed of information rigorously in a model of interacting bosons, more practical work has seemed to clearly confirm that physically relevant Bose gases do have a ``linear light cone" -- namely, a finite velocity with which quantum correlations and information can spread.  In fact, the first crisp experimental observation of a finite velocity of quantum correlations took place in an experiment on one-dimensional ultracold Bose gases \cite{nature12}.  Indeed, many authors \cite{Light08,barmettler,Light14,Light18,Takasueaba9255,kennett} have observed strict light cones in numerical simulations of these Bose gases, all while no rigorous results have been able to generalize the mathematically precise Lieb-Robinson bounds to interacting bosons.  (Of course, due to the challenge of proving a Lieb-Robinson bound for these models, one may not know with mathematical certainty that these simulations are guaranteed to have controllable error!).

This paper closes the longstanding gap between experiment and simulation on the one hand, and mathematical physics on the other.  We prove that correlation functions of interest in physical problems remain small outside of an emergent light cone which propagates with a finite velocity in ``thermal" states with infinite temperature, but a finite number density of bosons, in interacting boson models with density-dependent interactions on any lattice or graph.   In one dimensional models, we prove stronger results:  there is a finite velocity of quantum information in \emph{every} finite density state or ensemble.  As a consequence of this stronger 1d result, we also prove that simulating Bose-Hubbard-like models in 1d is not asymptotically more difficult than simulating a 1d model with a finite-dimensional Hilbert space. Similarly, like in models with finite-dimensional Hilbert spaces \cite{Hastings04,Hastings:2005pr,Nachtergaele_2006}, models with a gapped ground state have correlation functions (in said ground state) which exponentially decay with distance.   These results, along with the mathematical method we use to prove them (which differs somewhat from \cite{schuch,kuwahara2021liebrobinson}) form the key results of this paper. A schematic depiction of our result is provided in Figure \ref{fig:schematic}.

\begin{figure}[t]
\includegraphics[width=0.7\textwidth]{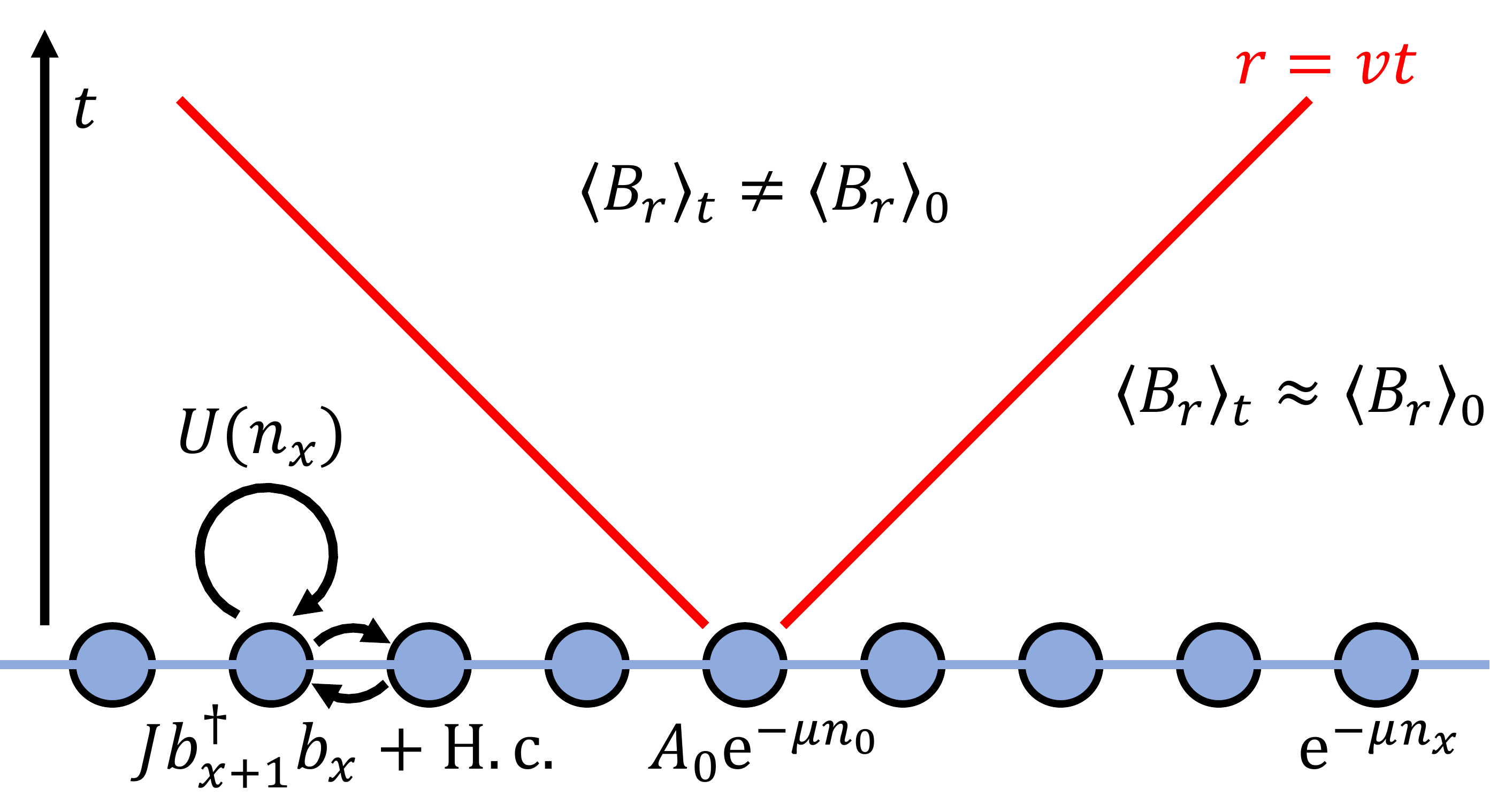}
\caption{\label{fig:schematic}  A schematic depiction of the linear light cone in a model of interacting bosons with single-body hopping terms, as in (\ref{eq:introBH}).   A local perturbation at the origin ($x=0$) can only affect expectation values of observables at position $x=r$ in a grand canonical, finite density ensemble, after a time $t\ge r/v$.  Our proof that the velocity $v$ is finite represents the first rigorous proof that quantum correlations and information must propagate with bounded velocities in a broad family of interacting bosonic models, including (but not limited to) the canonical Bose-Hubbard model.  }

\end{figure}

\section{Intuition behind our results and methods}
In this section, we present a non-rigorous overview of the key results, along with the mathematical techniques we use to prove them.  The following sections then contain the many technical details outlined here.

While our formal results are actually rather broad in scope, by far the most recognizable model to which they apply is the canonical Bose-Hubbard model \cite{gersch,giamarchi,fisher,jaksch,greiner}.  For pedagogical purposes, let us focus here on the one-dimensional version of this model, whose Hamiltonian is \begin{equation}
     H = \sum_{x=-\infty}^\infty \left(J b^\dagger_x b_{x+1} +  J b^\dagger_{x+1} b_{x} + Un_x(n_x-1)\right), \label{eq:introBH}
\end{equation}
where $b^\dagger_x$ and $b_x$ are bosonic creation and annihilation operators on site $x$, \begin{equation}
    n_x = b^\dagger_x b_x 
\end{equation}is the boson number operator, and \begin{equation}
    [b_x,b^\dagger_y] = \delta_{xy}.
\end{equation} 
In (\ref{eq:introBH}), we further assumed that the model is one-dimensional with nearest-neighbor hopping terms. 

Theorem \ref{thm51} proves that in bosonic models like this -- independent of the spatial dimensionality or other details of the lattice -- ``thermal" correlation functions in a finite density grand canonical ensemble are super-exponentially small outside of a ``linear light cone", just as they are in local spin chains.   If $A_0$ and $B_r$ represent two spatially local operators separated by distance $r$, and $\mathcal{O}(t) := \mathrm{e}^{\mathrm{i}Ht} \mathcal{O} \mathrm{e}^{-\mathrm{i}Ht}$ denotes Heisenberg time evolution of an operator,  \begin{equation}
    \frac{\mathrm{tr}\left(\mathrm{e}^{-\mu N} [B_r(t),A_0]\right)}{\mathrm{tr}\left(\mathrm{e}^{-\mu N}\right)} \le c \left(\frac{vt}{r}\right)^{c^\prime r}. \label{eq:main}
\end{equation}  
In this equation $c$ and $c^\prime$ are constants, $v$ is an upper bound on the ``speed of quantum information", and $\mu$ represents a chemical potential for the conserved number of bosons $N$.
\footnote{In a conventional statistical mechanics setting, one usually defines the density matrix as $\rho \propto \mathrm{e}^{-\beta H + \beta \mu_* N}$, with $\mu_*$ defined as the conventional chemical potential.   However, we will consider systems with infinite temperature, or $\beta=0$, where $\mu := -\beta \mu_*$ is held fixed.   Our sign convention on $\mu$ is also changed as it is far more convenient to have $\mu>0$ enforce a finite density of bosons.}
Note that $v$ and $c$ can depend on $\mu$, and in our bound can depend on the observables $A$ and $B$ as well (though this may be an artifact of our bound, and not a physical effect).   We emphasize that in this grand canonical ``thermal" ensemble, the number of bosons $N$ is macroscopically large: indeed, the average occupancy of bosons on a single site is \begin{equation}
    \langle n_x\rangle = \frac{1}{\mathrm{e}^\mu -1 } := \bar{n}.
\end{equation}
Our bound, which proves that $v$ is finite, holds for \emph{any} $0<\mu<\infty$, and thus any finite density $\bar n$.  This result provides a definitive negative answer to the question of whether physically realistic, number-conserving, models of interacting bosons can propagate quantum correlations and information infinitely fast in (typical) finite density states, and settles a decades-old problem in mathematical physics.

To motivate the form of (\ref{eq:main}), consider the following scenario.  We pick a random state at a given chemical potential $\mu$ (let's call it $|\psi\rangle$), and then apply a local perturbation to $|\psi\rangle$:  \begin{equation}
    |\psi^\prime \rangle := |\psi \rangle + \mathrm{i}\epsilon A_0|\psi\rangle + \cdots ,
\end{equation}
with $A_0$ a local operator.   We take the parameter $\epsilon$ to be small and real, and $A_0$ to be Hermitian, for pedagogical purposes here.  How much might this perturbation affect an observable $B_r$, located a distance $r$ away, by time $t$?   This is captured by \begin{align}
    \langle \psi^\prime(t)| B_r|\psi^\prime(t)\rangle - \langle \psi^\prime(0)| B_r|\psi^\prime(0)\rangle &=  \langle \psi| (1-\mathrm{i}\epsilon A_0^\dagger) \mathrm{e}^{\mathrm{i}Ht}    B_r \mathrm{e}^{-\mathrm{i}Ht} (1+\mathrm{i}\epsilon A_0)|\psi\rangle - \langle \psi| (1-\mathrm{i}\epsilon A_0^\dagger)     B_r  (1+\mathrm{i}\epsilon A_0)|\psi\rangle + \cdots \notag \\
    &= \langle \psi | \left(B_r(t)-B_r + \mathrm{i}\epsilon [B_r(t)-B_r, A_0]\right)|\psi\rangle + \cdots \notag \\
    &\approx \mathrm{i}\epsilon \langle \psi | [B_r(t),A_0]|\psi\rangle,
\end{align}
where in the last step we have assumed that $\langle \psi|B_r(t)|\psi\rangle$ is essentially time-independent (thus all time dependence arises entirely from our perturbation), and we have used that two operators which are spatially separated commute:  $[B_r,A_0]=0$.   Under time evolution, $B_r(t)$ becomes a highly non-local operator which can badly fail to commute with $A_0$.   (\ref{eq:main}) shows that the time $t$ required for this to happen is at least as large as $r/v$, for some finite velocity $v$.  

Let us now ascertain whether or not our bound has optimal scaling.  Assuming that operators $A$ and $B$ in (\ref{eq:main}) are creation or annihilation operators (e.g. $A_0 = b_0^\dagger$ and $B_r = b_x$ with $r=a\times x$ where $a$ represents the physical spacing between lattice sites), we find
\begin{equation}
    v \le (496 + 384 \bar n) \frac{Ja}{\hbar}. \label{eq:BHbound}
\end{equation}
Analytical and numerical studies of this particular model \cite{Light08,barmettler,Light14,Light18,Takasueaba9255,kennett} (albeit in studies of slightly different states or ensembles) suggest that \cite{barmettler} \begin{equation}
    v \lesssim (2 + 4 \bar n) \frac{Ja}{\hbar}, \label{eq:BHv}
\end{equation} 
with this bound believed to be tight both  when $\bar n \ll 1$ and when $\bar n \gg 1$.   In the former limit, the tightness of (\ref{eq:BHv}) is seen by noting that the bosonic problem is essentially non-interacting and the maximal velocity set by the dispersion relation of the hopping $J$ terms is $2Ja/\hbar$.  In the latter limit, one can justify the scaling $v\sim \bar n J $ by noting that in a high density state with strong interactions $U \gg J \bar n$, the boson creation/annihilation operators scale as $b,b^\dagger \sim \sqrt{\bar n}$.   Comparing our bound (\ref{eq:BHbound}) to (\ref{eq:BHv}), we see that it is around two orders of magnitude too large, but does capture the right \emph{scaling} of the density dependence both at high and low density.  Moreover, the functional form of our bound (\ref{eq:main}) is easily seen to be optimal by studying the hopping of even a single boson \cite{chen2019operator}.   As a consequence, our bound might be quantitatively, \emph{but not qualitatively}, improved.

As promised above, (\ref{eq:main}) holds in far more than simply the Bose-Hubbard model.  We will prove that our linear light cone bound remains valid for arbitrary spatially local density-dependent interactions, for time-dependent Hamiltonians, and with single-boson hopping terms on any mathematical graph (which of course includes physical lattices in one, two or three dimensions).

Let us briefly outline the steps required to obtain (\ref{eq:BHbound}), where they can be found in the paper, along with our broad strategy of proof.  In Section \ref{sec:models}, we will formally define the space of models which we study.  In Section \ref{sec:qw}, we will define a normalizable ``operator Hilbert space" for bosonic systems, where the grand canonical ensemble $\rho \sim \mathrm{e}^{-\mu N}$ is built into a natural inner product on this ``operator Hilbert space".  To motivate this construction, we first observe that time-dependence in (\ref{eq:main}) is most naturally phrased in the language of growing operators.  This suggests that, as in standard Lieb-Robinson approaches, it will be more natural to think of Heisenberg operator dynamics rather than Schr\"odinger state evolution.  However, a key shortcoming of studying operator dynamics -- and indeed, the critical challenge which has foiled prior attempts to derive bounds on bosonic models -- is that the natural operators of interest, such as $b_x$ and $b^\dagger_y$, is that: (\emph{1}) these operators are \emph{infinite dimensional} (since there are arbitrarily large numbers of bosons that can exist on each site), and (\emph{2}) even more alarmingly, these operators are unbounded.  Mathematically, we write $\lVert b_x\rVert=\infty$ -- the operator norm of $b_x$ does not exist.  Intuitively, this unboundedness just follows from the fact that, even for just one boson, $b|n\rangle = \sqrt{n}|n-1\rangle$, where $|n\rangle, |n-1\rangle$ are normalized: the coefficient $\sqrt{n}$ can be arbitrarily large.   To bound dynamics, we need to demonstrate that these $\sqrt{n}$ factors cannot contribute to ``dangerously fast" Heisenberg dynamics.  Given that prior numerics have already suggested $v\sim \bar n$, resolving this issue is not only techncial, but essential to understanding the physics of how locality might even be possible in a bosonic model.

The way that we overcome this technical challenge is to use the ``many-body quantum walk" formalism for operator growth \cite{Lucas:2019cxr,Yin:2020oze,Yin:2020pjd,Tran:2020xpc}.  In this approach, we take the ``operator Hilbert space" intuition seriously, and think about the operator $b_x(t)$ as a ``quantum state" in some new ``Hilbert space".  Since this new vector space is our own abstract construction, we might as well choose it cleverly, and in particular we find it convenient to define the following inner product between operators: \begin{equation}
    (A|B) := \mathrm{tr}\left(\sqrt{\rho} A^\dagger \sqrt{\rho}B\right), \label{eq:innerproduct}
\end{equation} 
with $\rho \propto \exp[-\mu N]$ for $0<\mu<\infty$.  The notation here is inspired by Dirac's bra-ket notation, but we use parentheses to emphasize that this ``Hilbert space" is not the physical one, but rather exists for operators.  The key feature of (\ref{eq:innerproduct}) is that states with a large number of bosons will have an exponentially small inner product. Therefore, we expect that the unboundedly fast quantum dynamics hinted at in the previous paragraph will be so suppressed by $\sqrt{\rho}$ that we can prove exact bounds on operator dynamics using this inner product.  

To get further intuition for this idea, obesrve that in the operator quantum walk, we write \cite{Chen:2020bmq}
\begin{equation} \label{eq:sec2cartoon}
    b_x(t) = \sum_{i} c_i(t) b_i + \sum_{ijk}c_{ij,k}(t) b_i b_j b^\dagger_k + \sum_{ijklm}c_{ijk,lm}(t) b_i b_j b_k b_l^\dagger b^\dagger_m  + \cdots.
\end{equation}
The coefficients $c_i(t)$, $c_{ij,k}(t)$, etc., are the coefficients of a ``quantum state", but the ``states" such as $b_i b_j b^\dagger_k$ are \emph{not normalized} in the inner product (\ref{eq:innerproduct}).  In fact, we could estimate that e.g. \begin{equation}
   \sqrt{ (b_i b_j b_k^\dagger | b_i b_j b_k^\dagger)} \sim \bar n^{3/2}.
\end{equation}
This means that, as we adjust the thermodynamic density $\bar n$ of interest, the \emph{same} Heisenberg operator $b_x(t)$ will be interpreted quite differently: when $\bar n \gg 1$, long operator strings will be more important than when $\bar n \ll 1$.  To get very rough insight into how this can give rise to a $\bar n$-dependent velocity (\ref{eq:BHv}), imagine that \begin{equation} \label{eq:roughbt}
    b(t) \sim \sum_{m=0}^\infty \frac{t^m}{m!} b_x \prod_{j=1}^m b_{x+j}b_{x+j}^\dagger.
\end{equation}
Since the length of each $b,b^\dagger \sim \sqrt{\bar n}$, we could estimate that the dominant term in the series above arose when $(\bar n t)^n/n!$ is maximal, or when $\bar n t \sim m$.  Since $m$ would correspond to the distance traveled, this would give us velocity $v\sim \bar n$.  In reality, the origin of (\ref{eq:BHv}) is a little more complicated in the Bose-Hubbard model, but this simple cartoon illustrates how a quantum walk formalism can crisply capture $\bar n$-dependent dynamics in an interacting boson model.

To actually \emph{prove} (\ref{eq:BHv}), note that non-vanishing commutators in (\ref{eq:main}) can only arise from the spatial growth of operators.  Therefore, we can actually bound (\ref{eq:main}) by carefully understanding how operator strings of $b$ and $b^\dagger$ evolve using the quantum walk.  To obtain exact results, we bound the growth of operator strings by defining well-chosen ``superobservables" $\mathcal{F}$ on the ``operator Hilbert space".  In a nutshell, we will choose \begin{equation}
    \mathcal{F} \sim \sum_{x=-\infty}^\infty \mathrm{e}^{\lambda |x|} \mathbb{P}_x,
\end{equation}
where $\mathbb{P}_x$ is a projection onto operator strings with at least one $b_x$ or $b^\dagger_x$, and then prove that (via Markov's inequality) \begin{equation}
    \text{if } ([A_0(t),B_x] | [A_0(t),B_x]) \sim 1, \; \text{then } (A_0(t)|\mathcal{F}_x|A_0(t)) \gtrsim \mathrm{e}^{\lambda x}.
\end{equation}
The precise implementation of this idea is detailed in Section \ref{sec:otoc}.


In Section \ref{sec:proof}, we then prove the linear light cone (\ref{eq:main}) by showing that \begin{equation} \label{eq:A0F}
    (A_0(t)|\mathcal{F}|A_0(t)) \lesssim \mathrm{e}^{\kappa t},
\end{equation}
for some finite constant $\kappa$.  This implies that the velocity in (\ref{eq:main}) is \begin{equation}
    v \le \frac{\kappa}{\lambda}.
\end{equation}
Intuitively, this is done by noting that with each step in time, the locality in $H$ means that $\mathcal{F}$ cannot increase too much.  A bit more precisely, we evaluate (\ref{eq:A0F}) in an interaction picture where the \emph{hopping} terms ($J$, in (\ref{eq:introBH})) in the Hamiltonian are treated as the ``perturbation", and the interactions ($U$, in (\ref{eq:introBH})) are the ``unperturbed" terms.  This is because we use a basis for ``operator Hilbert space" where the $U$-terms almost don't contribute to time dependence in $(A_0(t)| \mathcal{F} |A_0(t))$.  And if only hopping terms were present, a linear light cone would exist since the problem would reduce to a single-particle system where Lieb-Robinson bounds are well-established.  The large majority of our proof of the linear light cone amounts to characterizing the extent to which the interactions \emph{can} modify $(A_0(t)| \mathcal{F} |A_0(t))$.  The density-dependent interactions cause the accumulation of many powers of $b^\dagger b$ in (\ref{eq:sec2cartoon}), albeit all on the same lattice site.  It therefore becomes critical to carefully re-sum these contributions.  
Eventually these effects lead to an enhancement in $\kappa$, and hence the velocity of the light cone, beyond what the single-particle hopping terms alone could achieve.  Remarkably (\ref{eq:BHv}) shows this enhancement \emph{is a physical effect}.

The operator growth picture above immediately leads to both bounds on ordinary correlators such as $\mathrm{tr}(\rho [A_0(t),B_r])$, and bounds on out-of-time-ordered correlators $\mathrm{tr}(\sqrt{\rho} [A_0(t),B_r]\sqrt{\rho} [A_0(t),B_r])$: see Corollary \ref{cor45}.  Bounds on these correlators exist in any spatial dimension.  While our light cone is stronger than that in \cite{kuwahara2021liebrobinson}, our bound does not (as of now) apply to correlators in  the thermal state $\rho \sim \mathrm{e}^{-\beta H}$.

Our second main result is the proof of a much stricter notion of light cone in one dimensional models. Theorem \ref{thm71} proves that all matrix elements of $[A_0(t),B_r]$ between finite density quantum states are bounded by a light cone of the form (\ref{eq:main}).  This means that not only a typical finite density state, but \emph{all} finite density states, obey a linear light cone.  Intuitively, the proof of this result is straightforward.  In a chain of length $L$, the number of finite density states scales as $\mathrm{e}^{\mathrm{O}(L)}$.  In the worst case scenario, a bound on $\mathrm{tr}(\sqrt{\rho} [A_0(t),B_L]\sqrt{\rho} [A_0(t),B_L])$ is large because of a \emph{single matrix element} where the commutator is large.  Thus, given any density matrix $\tilde \rho$ corresponding  to a finite density state must have bounded entries: \begin{equation}
    \mathrm{tr}\left(\tilde \rho [A_0(t),B_L]\right) \lesssim \left(\mathrm{e}^{\mathrm{O}(L)}\right)^2 \mathrm{tr}(\sqrt{\rho} [A_0(t),B_L]\sqrt{\rho} [A_0(t),B_L]) \sim \left(\frac{\mathrm{O}(1) \cdot vt}{L}\right)^L.
\end{equation}
The super-exponential decay of (\ref{eq:main}) with $L$ is so strong that it allows us to safely salvage our bound: the number of states factor is fairly negligible.   However, we also need to modify the proof above to deal with the case where the two operators $A_0$ and $B_x$ are separated by distances $x\ll L$; given the picture of local operator growth sketched above, we are able to obtain this result with a bit of further work. 

We prove two important applications of this stronger bound in 1d models.  First,  we bound the classical computational complexity of simulating the Bose-Hubbard model, and prove in Section \ref{sec:complexity} that this task is asymptotically no harder in one dimension than simulating a local 1d spin chain.  This demonstrates a simple and practical application of our formal bound in condensed matter and atomic physics.  Second, we prove in Section \ref{sec:gs} that in any 1d interacting Bose gas with time-independent Hamiltonian and density-dependent interactions, correlation functions in the ground state $|E_0\rangle$ obey \begin{equation}
   \langle E_0| A_0 B_r|E_0\rangle - \langle E_0|A_0|E_0\rangle \langle E_0|B_r|E_0\rangle \lesssim \mathrm{e}^{-r/\xi} \label{eq:maings}
\end{equation}
whenever there is a finite energy gap to the first excited state.  (Here $\xi$ is a finite number, independent of $r$, and $A_0$ and $B_r$ denote two local operators separated by a distance $r$.)  The exponential decay with $r$ in (\ref{eq:maings}) is just as strong as it is in local models.   These two results rigorously show that at least in certain ways, models of interacting bosons -- despite their formally infinite dimensional Hilbert space -- can share many of the same physical properties as models of interacting spins or fermions.

The results highlighted above have many implications.  Here we highlight a few interesting ones, spanning atomic and condensed matter physics, together with quantum information. (\emph{1}) It is common when simulating a Bose-Hubbard model to truncate the Hilbert space, not allowing for arbitrarily large boson number fluctuations on any given site.  Our rigorous results can formally justify such assumptions; indeed, we will describe strong bounds on the computational complexity of classically simulating the Bose-Hubbard model in one dimension in Section \ref{sec:complexity}.  (\emph{2}) Section \ref{sec:gs} demonstrates that (at least in 1d) a simple feature of a phase of matter -- a gapped ground state -- will lead to a finite correlation length in correlation functions, independently of whether the local Hilbert space is finite or not.  Indeed, one would not expect such a mathematical detail to have a profound physical consequence, and our methods lead to a first rigorous demonstration of this expectation.  (\emph{3}) Our results demonstrate that it is not feasible to use Bose gases to asymptotically improve on the operating speed of a future quantum information processor: signals propagate at finite velocities in any physically realizable finite density state. Even though interactions can become arbitrarily strong if one engineers all of the bosons to clump together under the quantum dynamics, our result proves that these enhanced interactions cannot in fact form the basis for rapid spreading of quantum information or correlations.    



\section{Bosonic models with number conservation}\label{sec:models}
Let us now provide technical definitions of the models we will study in this paper. Consider an undirected graph $G=(V,E)$ with vertex set $V$ and edge set $E$ consisting of pairs of vertices.  We do not require $V$ or $E$ to be finite sets, but we will require that the degree of each vertex \begin{equation}
    \deg(v) = |\lbrace e \in E : v\in e \rbrace | \le K
\end{equation}for some finite number $K$; this simply means that each vertex has a finite number of neighbors. 

On each vertex, we place a single bosonic degree of freedom, corresponding to an infinite-dimensional Hilbert space spanned by the states $|n\rangle_v$ for $n\in \mathbb{Z}_{\ge 0}$.  The bosonic raising operator $b^\dagger_v$ and lowering operator $b_v$ on each site are defined as usual: \begin{subequations}\begin{align}
    b^\dagger_v |n\rangle_v &= \sqrt{n+1} |n+1\rangle_v, \\
    b_v |n\rangle_v &= \sqrt{n}|n-1\rangle_v.
\end{align}
\end{subequations}
The global Hilbert space $\mathcal{H}$ of the model contains all normalizable wave functions written in a product basis $\bigotimes_{v \in V} |n\rangle_v$.  Bosonic operators on different sites commute: \begin{equation}
    [b_u, b^\dagger_v] = \delta_{uv}.
\end{equation}
The number operator \begin{equation}
    n_v = b^\dagger_v b_v
\end{equation}
counts the number of bosons on vertex $v$.

In this paper, we will bound quantum dynamics generated by the time-dependent Hamiltonians of the generic form \begin{equation}
    H(t) = \sum_{\lbrace x,y\rbrace \in E} J_{xy}(t) b^\dagger_x b_y   + \sum_{S \subset V: \mathrm{diam}(S) \le \ell} U_S(n_{v \in S}, t)
\end{equation}
with $J_{xy}(t)$ a Hermitian matrix ($J_{xy} = \overline{J_{yx}}$, with overbar denoting complex conjugation), and $U_{S}(n_{v\in S},t)$ an \emph{arbitrary polynomial potential} in the density operators acting in a given subset $S\subset V$ with the property that all sites within $S$ are within a distance $\ell$ of each other.  Here the distance between vertices $u$ and $v$ is defined in the Manhattan sense -- the minimal number of edges traversed to get from one to the other. The dependence on $t$ in the Hamiltonian does not need to be continuous.  

The canonical example of such a model is the Bose-Hubbard model \cite{gersch}, in which after an appropriate choice of units for time: \begin{subequations}\begin{align}
    J_{xy}(t) &= 1, \\
    U_{\lbrace x\rbrace}(n,t) &= U_0 n(n-1) .
\end{align}\end{subequations}
with $U_0>0$ a constant.   However in this paper, the only requirement we will impose is that \begin{equation}
    J_{xy}(t) \le 1.
\end{equation}

A key property of these  models of interacting bosons is: 
\begin{prop}[Number conservation] \label{prop:numcons}
Let the total number of bosons be \begin{equation}
    N := \sum_{x\in V} b^\dagger_v b_v.
\end{equation}
Then \begin{equation}
    [N,H(t)] = 0.
\end{equation}
\end{prop}
This well-known result will be at the heart of our approach.  In particular, we will now describe a many-body quantum walk formalism which allows us to cleanly control the dynamics of ``thermal" correlators in a finite chemical potential grand canonical ensemble.

\section{Operator Hilbert space for bosons at finite density}\label{sec:qw}
Following \cite{Chen:2020bmq}, we now describe a many-body ``quantum walk formalism" for describing the growth of operators, and ultimately bounding thermal correlation functions.   Another approach which derived state-dependent commutator bounds can be found in \cite{Han:2018bqy}. We do so by defining the inner product (\ref{eq:innerproduct}) on the Hilbert space of operators, with $\rho$ the (grand canonical) thermal density matrix at infinite temperature and finite chemical potential $\mu$: \begin{equation}
    \rho = \bigotimes_{v\in V} (1-\mathrm{e}^{-\mu})\mathrm{e}^{-\mu n_v}. \label{eq:rho}
\end{equation}
We assume $0<\mu<\infty$.   We use the notation $|A)$, $|B)$ for operators to emphasize that the inner product space structure will be essential in the framework that follows.

If we were studying a single bosonic degree of freedom (graph $G$ has one vertex), a useful basis for operator Hilbert space would correspond to $\lbrace |n\rangle\langle n^\prime| : n,n^\prime \in \mathbb{Z}_{\ge 0} \rbrace$.
The Hilbert space of operators would consist of all states which have finite length: if \begin{equation}
    \mathcal{O} := \sum_{n,n^\prime=0}^\infty c_{nn^\prime} |n\rangle\langle n^\prime|,
\end{equation}
then \begin{equation}
    (\mathcal{O}|\mathcal{O}) = (1-\mathrm{e}^{-\mu}) \sum_{n,n^\prime=0}^\infty |c_{nn^\prime}|^2 \mathrm{e}^{-\mu (n+n^\prime)/2} < \infty.
\end{equation}
We will often use the notation \begin{equation}
    |nn^\prime) :=  \frac{\mathrm{e}^{\mu (n+n^\prime)/4}}{\sqrt{1-\mathrm{e}^{-\mu}}} |n\rangle\langle n^\prime|.
\end{equation}
The normalization constant is chosen so that these vectors are orthonormal:\begin{equation}
    (n_1n_1^\prime | n_2 n_2^\prime )  := \delta_{n_1n_2}\delta_{n_2n_2^\prime}.
\end{equation}
Note that in particular, the identity matrix \begin{equation}
    I := \sum_{n=0}^\infty |n\rangle \langle n|
\end{equation}
is a normalizable state and hence exists in the operator Hilbert space, so long as $\mu>0$: \begin{equation}
    |I) := \sqrt{1-\mathrm{e}^{-\mu}} \sum_{n=0}^\infty \mathrm{e}^{-\mu n/2} |nn). \label{eq:Inorm}
\end{equation}  We then define the projection superoperator \begin{equation}
    \mathbb{P}|\mathcal{O}) := |\mathcal{O}) - (I|\mathcal{O})|I) 
\end{equation}
to project any operator off of the identity, the projection operators \begin{equation}
    \mathbb{P}^{nn^\prime} = |nn^\prime)(nn^\prime|,
\end{equation}
and the ``identity superoperator" \begin{equation}
    \mathcal{I} := \sum_{n,n^\prime=0}^\infty \mathbb{P}^{nn^\prime}.
\end{equation}

Our choice of operator basis is a balancing act between two ``competing interests."  On the one hand, since $I$ commutes with all operators, it is ideal to separate out the identity, especially when bounding operator growth and the spreading of quantum information.   On the other, an operator basis such as $I,b,b^\dagger, b^\dagger b, \cdots$ turns out to be quite unwieldy.  Moreover, we will see that the  basis vectors $|nn^\prime)$ only pick up phases under time evolution under the density-dependent interactions $U_S$; this property will be particularly valuable in proving the light cone.  Ultimately, after some tinkering, we  found that working in the $|nn^\prime)$ operator basis, but with projecting out the identity, was the most effective strategy for describing growing operators that we could find.

Now, let us explain the straightforward generalization of this basis to a multi-site problem (vertex set $V$ now has more than one element).    We will typically use subscripts to denote that the objects defined above act on particular vertices:  for example, the projector off of operators that correspond to the identity on vertex $v$ is \begin{equation}
    \mathbb{P}_v := \underbrace{\mathbb{P}}_{\text{site }v} \otimes \underbrace{\bigotimes_{x\in V-v} \mathcal{I}}_{\text{other sites}}.
\end{equation}
Since $\rho$ is a tensor product between vertices, the inner product is well-behaved.   We will find it useful to define the projector onto operators which are not the identity on a subset $R\subset V$: \begin{equation}
    \mathbb{P}_R := 1 - \prod_{v\in R} (1-\mathbb{P}_v).  \label{eq:PR}
\end{equation}

We define the Liouvillian \begin{equation}
    \mathcal{L}(t) := \mathrm{i}[H(t),\cdot]
\end{equation}
to be a ``superoperator" (a linear transformation on the Hilbert space of operators).  The time evolution automorphism on this operator Hilbert space is defined by the equation \begin{equation}
    \frac{\mathrm{d}}{\mathrm{d}t} |A(t)) := \mathcal{L}(t)|A(t)). \label{eq:ddtA}
\end{equation}  We now state a number of useful formal properties of $\mathcal{L}$, and of this inner product space.

\begin{prop}
$\mathcal{L}(t)$ is anti-Hermitian: $\mathcal{L}^\dagger = -\mathcal{L}$, or  $(A|\mathcal{L}|B) = -\overline{(B|\mathcal{L}|A)}$ for any operators $A$ and $B$. \label{prop:antihermitian}
\end{prop}\begin{proof}
This result immediately follows from Proposition \ref{prop:numcons}: \begin{equation}
    (A|\mathcal{L}|B) = \mathrm{tr}\left(\sqrt{\rho} A^\dagger \sqrt{\rho} \mathrm{i}[H,B]\right) = \mathrm{tr}\left(\mathrm{i}\left[\sqrt{\rho} A^\dagger \sqrt{\rho}, H\right] B\right) = \mathrm{i} \times \overline{\mathrm{tr}\left(B^\dagger\left[H,\sqrt{\rho} A \sqrt{\rho}\right] \right)}=  \overline{\mathrm{tr}\left(-\mathrm{i} B^\dagger\sqrt{\rho} [H,A] \sqrt{\rho} \right)}
\end{equation}
where the second and third equalities follows from the cyclicity of the trace, and the fourth equality follows from the fact that for any operator $f(N)$, $[H,f(N)]=0$.
\end{proof}
From this result we immediately find the following useful results: 
\begin{cor}
Let $\mathcal{F}$ be a superoperator.  Then the expectation value of $\mathcal{F}$ in operator $|A(t))$ obeys the following equation: \begin{equation}
    \frac{\mathrm{d}}{\mathrm{d}t} (A(t)|\mathcal{F}|A(t)) = (A(t)| [\mathcal{F},\mathcal{L}(t)]|A(t)).
\end{equation} \label{cor:FL}
\end{cor}
\begin{proof}
This follows from (\ref{eq:ddtA}), and (by Proposition \ref{prop:antihermitian}) $(\mathcal{L}(t)|A(t))^\dagger = (A(t)|\mathcal{L}(t)^\dagger= -(A(t)|\mathcal{L}$.
\end{proof}

\begin{cor}
The length of states in operator Hilbert space does not change with time: \begin{equation}
    (A|A) = (A(t)|A(t)).
\end{equation}
\end{cor}

These three simple facts show us that it is possible to study operator growth in this system by thinking about $|A(t))$ as a normalizable quantum mechanical state in operator Hilbert space, undergoing a quantum walk.  Indeed, physical operators of interest such as $b_v$ and $b^\dagger_v$ are \emph{normalized} states in operator Hilbert space at any $\mu>0$:  for example, \begin{equation}
    |b_v) = \sum_{n=1}^\infty \sqrt{n}|n-1\rangle\langle n|_v = \sum_{n=1}^\infty \sqrt{n(1-\mathrm{e}^{-\mu})} \mathrm{e}^{-\mu (2n-1)/4}  |n-1,n)_v. \label{eq:bvnorm}
\end{equation}

\section{Bounding correlators and commutators} \label{sec:otoc}
In this section, our main purpose is to explain why the notion of normalizability in (\ref{eq:bvnorm}) is all that is required to bound thermal correlators.  We emphasize that it does \emph{not} matter that the conventional operator norm is unbounded.  In order to relate this quantum walk formalism to the questions most conventionally addressed in the literature, it is useful to introduce some auxiliary superoperators.  For simplicity, we start by working in the Hilbert space of a single boson -- as above, it will be straightforward to generalize using tensor products.
Define the superoperator \begin{equation}
    F^\beta = \sum_{n,n^\prime=0}^\infty \max(n+\beta,n^\prime+\beta)^\beta |nn^\prime)(nn^\prime|,
\end{equation}
together with \begin{equation}
    \mathcal{F}^\beta := \mathbb{P}F^\beta\mathbb{P}.
\end{equation}
The following technical proposition shows us the extent to which projecting onto or off of the identity can modify the operator weight in a given $|nn^\prime)$:
\begin{prop}\label{prop:projections}
On a single vertex, consider a normalizable operator \begin{equation}
    |\mathcal{O}) = \sum_{n,n^\prime=0}^\infty \mathcal{O}_{nn^\prime} |nn^\prime)
\end{equation}
obeying $(\mathcal{O}|\mathcal{O})=1$.   Then\begin{subequations}
    \begin{align}
        |(nn|1-\mathbb{P}|\mathcal{O})| &\le \sqrt{1-\mathrm{e}^{-\mu}}\mathrm{e}^{-\mu n/2} = (nn|I), \label{eq:projection1P} \\
        |(nn|\mathbb{P}|\mathcal{O})| &\le |\mathcal{O}_{nn}| + \sqrt{1-\mathrm{e}^{-\mu}}\mathrm{e}^{-\mu n/2}. \label{eq:projectionP} \\
        (I|F^\beta|I) &\le \beta^\beta (1-\mathrm{e}^{-\mu})^{-\beta} \label{eq:IFI}
    \end{align}
\end{subequations}
\end{prop}
\begin{proof}
Observe that since $(I|I)=1$, \begin{align}
        (nn|1-\mathbb{P}|\mathcal{O}) &= (nn|I)(I|\mathcal{O}) \le (nn|I) \sqrt{(I|I)(\mathcal{O}|\mathcal{O})} = (nn|I).
    \end{align}
    (\ref{eq:Inorm}) then gives us (\ref{eq:projection1P}), and (\ref{eq:projectionP}) then follows from the triangle inequality. For (\ref{eq:IFI}), \begin{equation}\label{eq:sum_na}
    (I|F^\beta|I) = (1-\mathrm{e}^{-\mu})\sum_{n= 0}^\infty \mathrm{e}^{-\mu n}(n+\beta)^\beta \le (1-\mathrm{e}^{-\mu})\frac{\beta^\beta}{\beta!} \sum_{n= 0}^\infty \mathrm{e}^{-\mu n}\frac{(n+\beta)!}{n!} =  \left(\frac{\beta}{1-\ee^{-\mu}}\right)^\beta,
\end{equation}
\end{proof}

The basic strategy for studying operator dynamics in the quantum walk formalism is to use Corollary \ref{cor:FL} to efficiently bound operator growth, by choosing a clever superoperator $\mathcal{F}^\beta$ which can constrain the correlation functions of interest.  Because bosonic operators are unbounded, some care is required in order to choose such a superoperator.  Luckily, the following proposition shows us that $\mathcal{F}^\beta$ is sufficient to bound the operator length of commutators:
\begin{prop}\label{prop33}
Let $R\subset V$, and define \begin{equation}\label{eq:O'}
    \mathcal{O}^\prime :=  \prod_{x \in R} \left(b^\dagger_x\right)^{\eta_x} b_x^{\zeta_x}.
\end{equation}
Then if \begin{subequations}\label{eq:beta_gamma}
    \begin{align}
        \beta &= \sum_{x\in R} (\eta_x+\zeta_x), \\
        \gamma &= \sum_{x\in R} (\eta_x-\zeta_x),
    \end{align}
\end{subequations}
we have the inequality
\begin{equation}
  ([\mathcal{O},\mathcal{O}^\prime]|[\mathcal{O},\mathcal{O}^\prime]) \le 8\beta^\beta \cosh\frac{\mu\gamma}{2} \left(1+\beta\left(\frac{\beta}{1-\mathrm{e}^{-\mu}}\right)^{\beta}\right) \times \sum_{x\in R}(\mathcal{O}|\mathcal{F}_x^\beta|\mathcal{O}) \label{eq:prop33} 
\end{equation}. 
\end{prop}
\begin{proof}
To avoid unnecessary clutter, in what follows we will typically drop the $\beta$ superscript on $\mathcal{F}$ below.  First, observe that since operators supported on disjoint sets commute, we may freely write \begin{equation}
    [\mathcal{O},\mathcal{O}^\prime ] = [\mathbb{P}_R\mathcal{O},\mathcal{O}^\prime ],
\end{equation}
with $\mathbb{P}_R$ defined in (\ref{eq:PR}).   Then, we apply the triangle inequality: \begin{equation}
    ([\mathbb{P}_R\mathcal{O},\mathcal{O}^\prime]|[\mathbb{P}_R\mathcal{O},\mathcal{O}^\prime]) \le 2 (\mathcal{O}^\prime (\mathbb{P}_R\mathcal{O}) | \mathcal{O}^\prime (\mathbb{P}_R\mathcal{O}) ) + 2((\mathbb{P}_R\mathcal{O})\mathcal{O}^\prime|  (\mathbb{P}_R\mathcal{O})\mathcal{O}^\prime ).  \label{eq:OOprimereverse}
\end{equation}
The analysis of each term is similar, so we focus on the first term.  Writing out \begin{equation}
   \mathbb{P}_R \mathcal{O} = \sum_{\mathbf{n}}\mathcal{O}_{\mathbf{n}}|\mathbf{n})
\end{equation}
where here and in the remainder of this paper, we will use $\mathbf{n}$ as a quick shorthand for ``all possible $|nn^\prime)_v$ on all vertices $v$", and defining $\mathbf{a}_u$ and $\mathbf{a}^\prime_u$ to be ``unit vectors" corresponding to $n_u=1$ or $n_u^\prime=1$ respectively (with all other components zero) we see that \begin{align}
    \mathcal{O}^\prime \mathbb{P}_R \mathcal{O} = \sum_{\mathbf{n}}\mathcal{O}_{\mathbf{n}} \mathrm{e}^{-\mu \gamma/4}  |\mathbf{n} + \mathbf{g}) \prod_{x \in R} \left(\prod_{j=1}^{\zeta_x} \sqrt{n_x+1-j} \times \prod_{k=1}^{\eta_x} \sqrt{n_x-\zeta_x + k}\right). \label{eq:412}
\end{align}
where \begin{equation}
    \mathbf{g} := \sum_{x\in R} (\eta_x-\zeta_x)\mathbf{a}_x.
\end{equation}
Note that we are being lazy about terms where $\zeta_x>n_x$, because there is a factor of 0 in the product above, so such terms will not be counted anyway. Now, observe that \begin{equation}
    \prod_{x \in R} \left(\prod_{j=1}^{\zeta_x} \sqrt{n_x+1-j} \times \prod_{k=1}^{\eta_x} \sqrt{n_x-\zeta_x + k}\right) \le \prod_{x\in R} (\sqrt{n_x+\eta_x})^{\zeta_x+\eta_x} \le \left(\beta + \sum_{x\in R} n_x\right)^{\beta/2}.  \label{eq:413}
\end{equation}
Combining (\ref{eq:412}) and (\ref{eq:413}), we see that \begin{equation}
    (\mathcal{O}^\prime (\mathbb{P}_R\mathcal{O}) | \mathcal{O}^\prime (\mathbb{P}_R\mathcal{O}) ) \le \sum_{\mathbf{n}} \left|\mathcal{O}_{\mathbf{n}}\right|^2 \mathrm{e}^{-\mu \gamma/2}\left(\beta + \sum_{x\in R} n_x\right)^{\beta} 
    . \label{eq:OprimeFgamma}
\end{equation}
Now, we will use a series of generally loose inequalities to simplify even further, and reduce this expectation value to sums over $(\mathcal{O}|\mathcal{F}_x|\mathcal{O})$.  Firstly, we observe that \begin{equation}\label{eq:sum_nx}
    \left(\beta + \sum_{x\in R} n_x\right)^\beta \le \beta^\beta \sum_{x\in R} (n_x+\beta)^\beta \le\beta^\beta \sum_{x\in R} \max(n_x+\beta,n_x^\prime+\beta)^\beta .
\end{equation}  Secondly, let us observe that $\mathbb{P}_R|\mathcal{O})$ is not the same as $\mathbb{P}_x|\mathcal{O})$, and therefore $(\mathcal{O}|\mathbb{P}_RF_x\mathbb{P}_R|\mathcal{O}) \ne (\mathcal{O}|\mathcal{F}_x|\mathcal{O})$.  However, we have the following proposition to handle this (we present a more general statement for later use).

\begin{prop}\label{propOFO}
Suppose $|\mathcal{O}) = \mathbb{P}_R|\mathcal{O})$, and let $|\tilde{\mathcal{O} }) = \mathbb{P}_v\mathbb{Q}|\mathcal{O}) + c (1-\mathbb{P}_v)\mathbb{Q}|\mathcal{O})$, where $c\in \mathbb{C}$, superoperator $\mathbb{Q} = \mathcal{I}_v \otimes \mathbb{Q}_{-v}$ is trivial on $v\in R$. Then
\begin{align}\label{eq:sum_nbp}
(\tilde{\mathcal{O}}|F_v|\tilde{\mathcal{O}}) =
\sum_{nn'}\max(n+\beta,n'+\beta)^\beta \lVert  \mathbb{P}_v^{nn^\prime} |\tilde{\mathcal{O}})\rVert_2^2 
    \le (2-\delta_{c=0})\lVert\mathbb{Q}\rVert^2 \left[ (\mathcal{O}|\mathcal{F}_v|\mathcal{O}) + |c|^2\left(\frac{\beta}{1-\mathrm{e}^{-\mu}}\right)^{\beta} (\mathcal{O}|\mathbb{P}_R|\mathcal{O})\right].
\end{align}
where we can further replace
\begin{align}\label{eq:PPF}
    (\mathcal{O}|\mathbb{P}_R|\mathcal{O}) \le \sum_{x\in R} (\mathcal{O}|\mathbb{P}_x|\mathcal{O}) \le \sum_{x\in R} (\mathcal{O}|\mathcal{F}_x|\mathcal{O})
\end{align}
\end{prop}
\begin{proof}
The triangle inequality implies that \begin{equation}
    \lVert  \mathbb{P}_v^{nn^\prime}|\tilde{\mathcal{O}})\rVert_2^2 
    \le  (2-\delta_{c=0}) \left(  \lVert \mathbb{P}_v^{nn^\prime}\mathbb{P}_v\mathbb{Q}|\mathcal{O})\rVert_2^2 + |c|^2 \lVert \mathbb{P}_v^{nn^\prime}(1-\mathbb{P}_v)\mathbb{Q}|\mathcal{O})\rVert_2^2\right) .\label{eq:propOFO}
\end{equation}
Using that \begin{equation}
   \lVert \mathbb{P}_v^{nn^\prime}\mathbb{P}_v\mathbb{Q}|\mathcal{O})\rVert_2 = \lVert\mathbb{Q} \mathbb{P}_v^{nn^\prime}\mathbb{P}_v|\mathcal{O})\rVert_2 \le \lVert \mathbb{Q}\rVert \lVert \mathbb{P}_v^{nn^\prime}\mathbb{P}_v|\mathcal{O})\rVert_2, \label{eq:commuteQ}
\end{equation} 
the first term on the right hand side of (\ref{eq:propOFO}) after summed over $n,n'$ is bounded by $(2-\delta_{c=0})\lVert \mathbb{Q}\rVert^2 (\mathcal{O}|\mathcal{F}_v|\mathcal{O})$.   For the second term, we analogously pull out the factor $\lVert \mathbb{Q}\rVert$ and then use (\ref{eq:IFI}) and (\ref{eq:Inorm}). Suppose $R=\{x_i:i=1,\cdots,|R|\}$, (\ref{eq:PPF}) comes from
\begin{align}
    \mathbb{P}_R = \sum^{|R|}_{i=1}\mathbb{P}_{x_i} \prod^{i-1}_{j=1} (1-\mathbb{P}_{x_j})
\end{align}
and $\lVert 1-\mathbb{P}_{x}\rVert \le 1 \le \lVert F_x\rVert $.
\end{proof}

Combining (\ref{eq:OprimeFgamma}), (\ref{eq:sum_nx}) and Proposition \ref{propOFO} with $c=1, \mathbb{Q}=\mathcal{I}$, we obtain \begin{align}
    (\mathcal{O}^\prime (\mathbb{P}_R\mathcal{O}) | \mathcal{O}^\prime (\mathbb{P}_R\mathcal{O}) ) &\le 2\beta^\beta \mathrm{e}^{-\mu\gamma/2} \sum_{x\in R}\left[ (\mathcal{O}|\mathcal{F}_x|\mathcal{O}) + \left(\frac{\beta}{1-\mathrm{e}^{-\mu}}\right)^{\beta} \sum_{y\in R} (\mathcal{O}|\mathcal{F}_y|\mathcal{O})\right] \notag \\
    &\le 2\beta^\beta \mathrm{e}^{-\mu\gamma/2} \left(1+\beta\left(\frac{\beta}{1-\mathrm{e}^{-\mu}}\right)^{\beta}\right)\sum_{y\in R} (\mathcal{O}|\mathcal{F}_y|\mathcal{O}). \label{eq:422}
\end{align}

Bounding $( (\mathbb{P}_R\mathcal{O})\mathcal{O}^\prime |  (\mathbb{P}_R\mathcal{O})\mathcal{O}^\prime )$ requires analogous steps, but with $n^\prime_x$ replacing $n_x$ in the intermediate equalities, and with a factor of $\mathrm{e}^{\mu\gamma/2}$ instead of $\mathrm{e}^{-\mu\gamma/2}$: \begin{equation}
    ( (\mathbb{P}_R\mathcal{O})\mathcal{O}^\prime |  (\mathbb{P}_R\mathcal{O})\mathcal{O}^\prime ) \le 2\beta^\beta \mathrm{e}^{\mu\gamma/2} \left(1+\beta\left(\frac{\beta}{1-\mathrm{e}^{-\mu}}\right)^{\beta}\right)\sum_{y\in R} (\mathcal{O}|\mathcal{F}_y|\mathcal{O}). \label{eq:423}
\end{equation}
Combining (\ref{eq:OOprimereverse}), (\ref{eq:422}) and obtain (\ref{eq:prop33}).
\end{proof}

We emphasize that especially for $\beta>1$, the coefficients in (\ref{eq:prop33}) are not tight.  Nevertheless, they are sufficient to prove a linear light cone in bosonic models with super-exponentially small tails, which is the main purpose of this paper.  Indeed Proposition \ref{prop33} will lie at the heart of our proof of a linear light cone, since we will show how to use the quantum walk formalism to bound $(\mathcal{O}(t)|\mathcal{F}_v|\mathcal{O}(t))$. 
Note that Proposition \ref{prop33} does not restrict the form of $\mathcal{O}$ apart from normalizability, and easily generalizes to operators $\mathcal{O}^\prime$ beyond strings of $b,b^\dagger$, as long as its expansion coefficients on the basis $|n\rangle\langle n^\prime|$ are bounded by a polynomial of $n,n'$. The linear light cone result in the next section naturally follows for such generalized operators.

Our next goal is to explain how Proposition \ref{prop33} is also strong enough to constrain physically relevant correlation functions of interest.  Usually, the physical operators $A$ of interest obey $[A,N]=kA$ for some $k\in\mathbb{Z}$; this holds for example if $A$ is any product of creation and annihilation operators. On such products (or sums thereof), our inner product is easily related to more conventional thermal expectation values: 
\begin{prop}
If \begin{subequations}\label{eq:ANBN}\begin{align}
    [A,N] &= (k+k^\prime)A, \\
    [B,N] &= kB,
\end{align}\end{subequations}
then for any $t_A,t_B\in\mathbb{R}$, \begin{equation}
    (A(t_A)|B(t_B)) = \delta_{k^\prime,0} \mathrm{e}^{\mu k/2} \mathrm{tr}\left(\rho A(t_A)^\dagger B(t_B)\right). \label{eq:prop34}
\end{equation}
\end{prop}
\begin{proof}
Using Proposition \ref{prop:numcons}, and letting $U_B$ be the time evolution operator for time $t_B$, \begin{equation}
    [N,B(t_B)] = [N,U_B^\dagger BU_B] = U_B^\dagger [N,B]U_B = -k B(t_B).
\end{equation}
In the last step we used (\ref{eq:ANBN}).  For this reason, we can without loss of generality (and for ease of notation) set $t_A=t_B=0$, since our results do not depend on time evolution. Now let $|\psi_M\rangle$ denote an eigenvector of $N$ with eigenvalue $M$, and consider that due to (\ref{eq:ANBN}), \begin{equation}
    NB|\psi_M\rangle = B(N-k)|\psi_M\rangle = (M-k)B|\psi_M\rangle. \label{eq:BNeig}
\end{equation}More generally,
\begin{equation}
   A^\dagger \sqrt{\rho} B \sqrt{\rho}|\psi_M\rangle = A^\dagger \sqrt{\rho} B \mathrm{e}^{-\mu M/2}|\psi_M\rangle = \mathrm{e}^{-\mu M/2}A^\dagger \sqrt{\rho} B|\psi_M\rangle = \mathrm{e}^{-\mu M/2} A^\dagger B \mathrm{e}^{-\mu (M-k)/2}|\psi_M\rangle. \label{eq:moverho}
\end{equation}
Observe that this final state is an eigenvector of $N$ with eigenvalue $M-k+(k+k^\prime)=M+k^\prime$, analogously to (\ref{eq:BNeig}).

Now if we wish to evaluate\begin{equation}
    \mathrm{tr}(A^\dagger \sqrt{\rho}B \sqrt{\rho}) = \sum_{M=0}^\infty \sum_{|\psi_M\rangle} \langle \psi_M | A^\dagger \sqrt{\rho}B\sqrt{\rho}|\psi_M\rangle ,
\end{equation}
we observe that the trace can be evaluated as a sum over all possible states with a fixed number of bosons $M$.  Clearly, this inner product can only be non-zero if $k^\prime=0$.  Moreover, using (\ref{eq:moverho}), we can easily write \begin{equation}
    \mathrm{tr}(A^\dagger \sqrt{\rho}B \sqrt{\rho}) = \mathrm{tr}\left(A^\dagger B \rho \right) \mathrm{e}^{\mu k/2},
\end{equation}
which is equivalent to (\ref{eq:prop34}).
\end{proof}
Using the Cauchy-Schwarz inequality, we immediately see that:
\begin{cor}\label{cor45}
Suppose that for any fixed $\epsilon>0$, there exists a velocity $v$ such that for two vertices $x,y\in V$ separated by distance $r$, for $t < r/v$, \begin{equation}
    ([\mathcal{O}_x(t),\mathcal{O}^\prime_y]|[\mathcal{O}_x(t),\mathcal{O}^\prime_y]) \le \epsilon.
\end{equation}
 Then there also exist constants $\epsilon^\prime$ and $\epsilon^{\prime\prime}$ such that the following inequalities hold: 
 \begin{subequations}\begin{align}
    \mathrm{tr}\left(\rho [\mathcal{O}_x(t),\mathcal{O}^\prime_y]\right) &< \epsilon^\prime, \\
    \mathrm{tr}\left(\rho [\mathcal{O}_x(t),\mathcal{O}^\prime_y]^\dagger [\mathcal{O}_x(t),\mathcal{O}^\prime_y]\right) &< \epsilon^{\prime\prime}.
\end{align}\end{subequations}
Therefore, there is also a finite velocity $v$ at which correlations spread in ordinary thermal correlators.
\end{cor}

\section{Linear light cone}\label{sec:proof}
We are now ready to state our main result, which amounts to the rigorous statement and proof of (\ref{eq:main}).
\begin{thm}[Finite speed of correlations]
Let $\mathcal{O}$ denote an operator with initial support on the subset $R\subset V$: namely, $(1-\mathbb{P}_{R^{\mathrm{c}}})|\mathcal{O}) = |\mathcal{O})$.   Let operator $\mathcal{O}^\prime$ have support in subset $S\subset V$.  Suppose that for all vertices $u\in R$ and $v \in S$, $\mathrm{dist}(u,v)\ge r$; we denote this with $\mathrm{dist}(R,S)=r$.  Then \begin{equation}
    ([\mathcal{O}(t),\mathcal{O}^\prime] | [\mathcal{O}(t),\mathcal{O}^\prime] ) \le C \times \left(\frac{vt}{r}\right)^{r/(2\ell+1)}, \label{eq:thm1}
\end{equation}
for $v|t|<r$, where \begin{equation}
    C = 16\beta^\beta \cosh\frac{\mu\gamma}{2} \left(1+\beta\left(\frac{\beta}{1-\mathrm{e}^{-\mu}}\right)^{\beta}\right)\times \left[\sum_{x\in R} (\mathcal{O}|F_x^\beta|\mathcal{O}) + \left(\frac{\beta}{1-\mathrm{e}^{-\mu}}\right)^{\beta} (|R|+|R_\ell|)(\mathcal{O}|\mathcal{O})\right],
\end{equation} $R_\ell=\{x\in V: \mathrm{dist}(x,R)\le \ell\}$, $\beta$ and $\gamma$ are defined in Proposition \ref{prop33} based on the properties of $\mathcal{O}^\prime$, time evolution is generated by a Hamiltonian $H(t)$ obeying the constraints described in Section \ref{sec:models}, and the velocity \begin{equation}
    v < \left\lbrace\begin{array}{ll} 8K(31+24\mu^{-1}) &\ \beta=1, \ell = 0 \\ 92K(2\beta)^{\beta+1}(1+2\mu^{-1})^{\beta+1} &\ \beta>1, \ell=0 \\ 2^{\beta+10}(2l+1)K^{3\ell+2}\beta^{2\beta}(1+2\mu^{-1})^{2\beta} &\ \ell>0  \end{array}\right.. \label{eq:thm2}
\end{equation}\label{thm51}
\end{thm}
\begin{proof}
The proof of this result follows the general strategy of previous quantum walk based proofs on quantum information dynamics (e.g. \cite{Lucas:2019cxr,Tran:2020xpc,Yin:2020pjd}).   We will show that \begin{equation}
    (\mathcal{O}(t)|\mathcal{F}_x|\mathcal{O}(t)) \le C_x(t) \label{FCbound}
\end{equation}for each vertex $x\in V$, where the functions $C_x(t)$ obey the differential equations
\begin{equation}
    \frac{\mathrm{d} C_u}{\mathrm{d}t}  \le \sum_{u \in V : \mathrm{dist}(u,v) \le 1+\ell} M_{uv}(t) C_v(t) \label{eq:dCudt}
\end{equation}
subject to appropriate initial conditions on the $C_v(t)$, which we will shortly explain.
Finding bounds on the coefficients $M_{uv}(t)$ is somewhat tedious, and will take up much of the proof of this overall theorem.   Once we have a bound on $M_{uv}(t)$, we will integrate this differential equation to find a bound on $(\mathcal{O}(t)|\mathcal{F}_x|\mathcal{O}(t))$. Proposition \ref{prop33} will then complete the proof.

Let us now carry out these steps.  The first step is to provide a useful definition for $C_v(t)$.
In order to prove this result, we will use an interaction picture similar to \cite{kuwahara2021liebrobinson}.  Let us denote with $\mathcal{L}_J$ and $\mathcal{L}_U$ the Liouvillians corresponding to the $J$ and $U$ terms in the Hamiltonian respectively.  Letting $\mathcal{T}$ denote the time-ordering operator, we define \begin{equation}
    \mathcal{L}_J(t)_U := \mathrm{i}[H_J(t)_U,\cdot]
\end{equation}
where \begin{equation}
    H_J(t)_U := \mathcal{T}\exp\left[ \int\limits_0^t \mathrm{d}t^\prime \mathcal{L}_U(t^\prime)\right] H_J(t)
\end{equation}
is the interaction picture hopping term.

The key observation is that $U$ is a sum of mutually commuting operators, which means that we may write \begin{equation}
    \mathcal{T}\exp\left[ \int\limits_0^t \mathrm{d}t^\prime \mathcal{L}_U(t^\prime)\right] = \prod_{S\subset V} \mathcal{T}\exp\left[ \int\limits_0^t \mathrm{d}t^\prime \mathcal{L}_{U,S}(t^\prime)\right] \label{eq:LVprod}
\end{equation}
where $\mathcal{L}_{U,S} = \mathrm{i}[U_S,\cdot]$; the ordering of the product above does not matter.   So, this means that \begin{equation}
    H_{J,uv}(t)_U = \prod_{S : \lbrace u,v\rbrace \cap S \ne \emptyset} \mathcal{T}\exp\left[ \int\limits_0^t \mathrm{d}t^\prime \mathcal{L}_{U,S}(t^\prime)\right] H_{J,uv}(t).
\end{equation}
Observe that this operator is the identity on any site which is farther than $\ell+1$ sites away from either $u$ or $v$.   Let us denote with \begin{equation}
    \mathcal{B}_{uv} := \lbrace y \in V :  \min(\mathrm{dist}(y,u),\mathrm{dist}(y,v)) \le \ell \rbrace. \label{eq:defBuv}
\end{equation}
Then letting $\mathbf{n}_{\mathcal{B}_{uv}}$ denote only the occupation numbers for sites in $\mathcal{B}_{uv}$, we may write \begin{equation}
    H_{J,uv}(t)_U = I_{\mathcal{B}_{uv}^{\mathrm{c}}} \otimes J_{uv}(t) \sum_{\mathbf{n}_{\mathcal{B}_{uv}} } \sqrt{n_v(n_u+1)} |\mathbf{n}_{\mathcal{B}_{uv}} + \mathbf{a}_u-\mathbf{a}_v\rangle \langle \mathbf{n}_{\mathcal{B}_{uv}}| \times \mathrm{e}^{\mathrm{i}\theta(\mathbf{n}_{\mathcal{B}_{uv}},t)} + \mathrm{H.c.} \label{eq:HJVt}
\end{equation}
To derive this result, we have used that the interactions $H_U$ are diagonal in the occupation number basis, and hence only contribute an overall phase to the operator: \begin{equation}
    \theta(\mathbf{n},t) := \int\limits_0^t \mathrm{d}t^\prime \left[ U(\mathbf{n}+\mathbf{a}_u-\mathbf{a}_v, t^\prime) - U(\mathbf{n},t^\prime) \right]. \label{eq:deftheta}
\end{equation}
In this equation we are using the diagonal elements of the operators $U$, using the expected notation.   The key observation about (\ref{eq:HJVt}) is that the operators are almost the same as single boson hopping operators, except for the possibility of an arbitrary phase factor.   However, this phase factor will be mild and possible to account for in what follows.

Next, we write \begin{align}
   |\mathcal{O}(t)) = \mathcal{T}\exp\left[\int\limits_0^t \mathrm{d}t^\prime \mathcal{L}_J(t^\prime)_U\right] \times  \mathcal{T}\exp\left[\int\limits_0^t \mathrm{d}t^\prime \mathcal{L}_U(t^\prime)\right] |\mathcal{O}) &:= \mathcal{T}\exp\left[\int\limits_0^t \mathrm{d}t^\prime \mathcal{L}_J(t^\prime)_U\right]  |\mathcal{O}(t)_U) \notag \\
   &:=  \mathcal{U}(t)  |\mathcal{O}(t)_U) ,
\end{align}
and observe that \begin{equation}
    (\mathcal{O}(t)|\mathcal{F}_x|\mathcal{O}(t)) = (\mathcal{O}(t)_U| \mathcal{U}(t)^\dagger \mathcal{F}_x \mathcal{U}(t)  |\mathcal{O}(t)_U).
\end{equation}
We will then choose our initial conditions $C_x(0)$ such that \begin{equation}
    C_x(0) \ge (\mathcal{O}(t)_U|\mathcal{F}_x|\mathcal{O}(t)_U), \;\;\; \text{for any }t, \label{eq:Cx0}
\end{equation}
and will choose the $M_{uv}(t)$ such that \begin{equation}
    \left(\mathcal{O}\left| \left[\mathcal{F}_x, \mathcal{L}_J(t)_U\right]\right|\mathcal{O}\right) \le \sum_{y \in V} M_{xy}(t) (\mathcal{O}|\mathcal{F}_y|\mathcal{O}), \;\;\; \text{for all } |\mathcal{O}). \label{eq:Mdefine}
\end{equation}
If we can achieve (\ref{eq:Cx0}) and (\ref{eq:Mdefine}), then we will obtain (\ref{FCbound}) and (\ref{eq:dCudt}).   We will obtain each of these two desired results in turn. 

\begin{lem}\label{lem52} Suppose the operator $|\mathcal{O})$ is supported in an initial set $R$: if $R^{\mathrm{c}}$ denotes the complement of $R$, then \begin{equation}
    |\mathcal{O}) = (1-\mathbb{P}_{R^{\mathrm{c}}})|\mathcal{O}).
\end{equation}
Then (\ref{FCbound}) holds if we choose \begin{equation}
    C_x(0) = \left\lbrace\begin{array}{ll} 2\beta^\beta(1-\mathrm{e}^{-\mu})^{-\beta} + 2(\mathcal{O}|F_x|\mathcal{O}) &\ x \in R \\ 4\beta^\beta(1-\mathrm{e}^{-\mu})^{-\beta} &\ 0<\mathrm{dist}(x,R)\le \ell \\ 0 &\ \mathrm{otherwise} \end{array}\right..  \label{eq:lem42}
\end{equation}
\end{lem}
\begin{proof}
We begin by writing the operator \begin{equation}
    |\mathcal{O}) = \left(\sum_{\mathbf{n}_R} \mathcal{O}_{\mathbf{n}_R}|\mathbf{n}_R) \right) \otimes \bigotimes_{y \in R^{\mathrm{c}} } |I)_y.
\end{equation}
Due to (\ref{eq:LVprod}), $|\mathcal{O}(t)_U)$ remains to be the identity $I$ on $x$ for $\mathrm{dist}(x,R)>\ell$, thus $C_x(0)=0$ in this case. For $\mathrm{dist}(x,R)\le \ell$, using Proposition \ref{prop:projections} and the fact that interaction does not grow size $n,n'$, we have  \begin{align}
   \left\lVert \mathbb{P}_x^{nn^\prime}\mathbb{P}_x |\mathcal{O}(t)_U)\right\rVert_2^2 &\le \left(\left\lVert \mathbb{P}_x^{nn^\prime} |\mathcal{O}(t)_U)\right\rVert_2 + \sqrt{1-\mathrm{e}^{-\mu}}\mathrm{e}^{-\mu n/2} \delta_{nn^\prime}\lVert\mathcal{O}\rVert_2\right)^2 \le 2\left\lVert \mathbb{P}_x^{nn^\prime} |\mathcal{O})\right\rVert_2^2 + 2(1-\mathrm{e}^{-\mu})\mathrm{e}^{-\mu n} \delta_{nn^\prime}(\mathcal{O}|\mathcal{O}) \label{eq:PPbound}
\end{align}

\comment{
we can explicitly expand this expression out.  Letting $n_z$ denote the boson occupancy number on any site $z$ obeying $0<\mathrm{dist}(z,x)\le \ell$, we find \begin{align}
   \mathcal{T}\exp\left[ \int\limits_0^t \mathrm{d}t^\prime \mathcal{L}_U(t^\prime)\right] |\mathcal{O}) &= \left(\sum_{n,n^\prime=0}^\infty \sum_{\text{for all }z, n_z=0}^\infty |n_xn^\prime_x, n_zn_z) \mathcal{O}_{nn^\prime} \mathrm{e}^{\mathrm{i}\theta(\mathbf{n},t)}  \prod_{z:0<\mathrm{dist}(z,x)\le \ell}\sqrt{1-\mathrm{e}^{-\mu}}\mathrm{e}^{-\mu n_z/2}\right) \notag \\
   &\otimes \bigotimes_{y : \mathrm{dist}(y,x) > \ell} |I)_y
\end{align}
where we have used (\ref{eq:Inorm}) to compute the weights of the operators in the first line above, and $\theta$ is given in (\ref{eq:deftheta}).}

Then \begin{align}
    (\mathcal{O}(t)_U|\mathcal{F}_x|\mathcal{O}(t)_U) &= (\mathcal{O}(t)_U|\mathbb{P}_x  F_x \mathbb{P}_x|\mathcal{O}(t)_U) = \sum_{n,n^\prime=0}^\infty \max(n+\beta,n^\prime+\beta)^\beta \left\lVert \mathbb{P}_x^{nn^\prime}\mathbb{P}_x |\mathcal{O}(t)_U)\right\rVert_2^2 \notag\\ 
    &\le 2(\mathcal{O}|F_x|\mathcal{O})+ 2(\mathcal{O}|\mathcal{O})\sum_{n=0}^\infty (n+\beta)^\beta (1-\mathrm{e}^{-\mu}) \mathrm{e}^{-\mu n} \le 2(\mathcal{O}|F_x|\mathcal{O})+ 2\left(\frac{\beta}{1-\mathrm{e}^{-\mu}}\right)^\beta(\mathcal{O}|\mathcal{O}), \label{eq:PP}
\end{align}
where for $0<\mathrm{dist}(x,R)\le \ell$ we can further simplify using $(\mathcal{O}|F_x|\mathcal{O}) = (I|F|I)$ and (\ref{eq:IFI}).
\end{proof}

The next step is to derive (\ref{eq:Mdefine}), which we achieve using the following lemma: \begin{lem}\label{lem53}
(\ref{eq:Mdefine}) holds with \begin{equation}
    M_{uv}(t) \le \delta_{\mathrm{dist}(u,v)\le 2\ell+1} \times \left\lbrace\begin{array}{ll} 62+48\mu^{-1} &\ \ell=0,\beta=1 \\ 23(2\beta)^{\beta+1}(1+2\mu^{-1})^{\beta+1} &\ \ell=0,\beta>1 \\ 2^{\beta+8}\beta^{2\beta}(1+2\mu^{-1})^{2\beta}K^{l+1} &\ \ell>0 \end{array}\right., \;\;\; (u\ne v) \label{eq:Muv}
\end{equation} and \begin{equation}
    M_{uu}(t) \le \left\lbrace\begin{array}{ll} (62+48\mu^{-1})K &\ \ell=0,\beta=1 \\ 23(2\beta)^{\beta+1}(1+2\mu^{-1})^{\beta+1}K &\ \ell=0,\beta>1 \\ 2^{\beta+8}\beta^{2\beta}(1+2\mu^{-1})^{2\beta} K^{l+1} &\ \ell>0 \end{array}\right.
\end{equation}
\end{lem}
\begin{proof}
The proof of this result is somewhat tedious, and the reader may wish to skim or skip this part (or only read a subset to get the general idea).  In a nutshell, we simply need to expand out \begin{align}
    (\mathcal{O}|[\mathcal{F}_z,\mathcal{L}_J(t)_U] |\mathcal{O}) &= (\mathcal{O}|[\mathbb{P}_z F_z \mathbb{P}_z,\mathcal{L}_J(t)_U] |\mathcal{O}) \notag \\
    &= (\mathcal{O}|\mathbb{P}_z[F_z,\mathcal{L}_J(t)_U]\mathbb{P}_z |\mathcal{O})+ (\mathcal{O}|[\mathbb{P}_z,\mathcal{L}_J(t)_U]F_z\mathbb{P}_z |\mathcal{O}) + (\mathcal{O}|\mathbb{P}_zF_z[\mathbb{P}_z,\mathcal{L}_J(t)_U] |\mathcal{O})  \notag \\
    &= (\mathcal{O}|\mathbb{P}_z[F_z,\mathcal{L}_J(t)_U]\mathbb{P}_z |\mathcal{O})+2(\mathcal{O}|\mathbb{P}_zF_z[\mathbb{P}_z,\mathcal{L}_J(t)_U] |\mathcal{O}).
\end{align}
The third line follows from Proposition \ref{prop:antihermitian}, and from the Hermiticity of superoperators $F_z$ and $\mathbb{P}_z$.  In what follows, to avoid clutter, we will simply write $\mathcal{L}_{uv} = \mathcal{L}_{J,uv}(t)_U$ and $\mathcal{L}_J=\mathcal{L}_J(t)_U$.  Since $\mathbb{P}_z$ is a projector, we have \begin{equation}
    [\mathbb{P}_z,\mathcal{L}_J] = \sum_{uv\in E: \mathrm{dist}(z, \lbrace u,v\rbrace )\le \ell} [\mathbb{P}_z,\mathcal{L}_{uv}]= \sum_{uv\in E: \mathrm{dist}(z, \lbrace u,v\rbrace )\le \ell} \left[\mathcal{L}_{uv}(1-\mathbb{P}_z) - (1-\mathbb{P}_z)\mathcal{L}_{uv}\right].
\end{equation}
So ultimately, we need to evaluate \begin{align}
    (\mathcal{O}|[\mathcal{F}_z,\mathcal{L}_{uv}] |\mathcal{O}) &=(\mathcal{O}|\mathbb{P}_z[F_z,\mathcal{L}_{uv}]\mathbb{P}_z |\mathcal{O})+2(\mathcal{O}|\mathbb{P}_zF_z\mathcal{L}_{uv} (1-\mathbb{P}_z) |\mathcal{O})-2(\mathcal{O}|\mathbb{P}_zF_z(1-\mathbb{P}_z)\mathcal{L}_{uv}  |\mathcal{O})
\end{align} We will call the terms above ``case 1", ``case 2" and ``case 3" respectively, and will evaluate each in term.   For cases 2 and 3, we also need to handle separately the possibility that $z \in \lbrace u,v\rbrace$ (case A) or $z\notin \lbrace u,v\rbrace$ (case B).   In what follows, we will also use the notation
\begin{subequations} \begin{align}
    |\overline{\mathcal{O}}_z) &:= (1-\mathbb{P}_z)|\mathcal{O}), \\
    |\tilde{\mathcal{O}}_z) &:= \mathbb{P}_z|\mathcal{O}).
    \end{align}
\end{subequations}
Lastly, we will use the fact that, since operators supported on disjoint sets commute,
\begin{equation}
\mathcal{L}_{uv} |\mathcal{O}) =     \mathcal{L}_{uv} \mathbb{P}_{\mathcal{B}_{uv}}  |\mathcal{O}) = \mathbb{P}_{\mathcal{B}_{uv}}  \mathcal{L}_{uv} \mathbb{P}_{\mathcal{B}_{uv}}  |\mathcal{O}). \label{eq:PBuv}
\end{equation}
However, to avoid clutter, we will often not bother to write $\mathbb{P}_{\mathcal{B}_{uv}}$ explicitly, except where necessary or useful.

\textbf{Case 1:} Since $\mathcal{L}_{uv}$ only grows size $n,n'$ on site $u,v$, we only need to consider the case $z=u$. First rearrange the projectors
\begin{align}
    (\mathcal{O}|\mathbb{P}_u[F_u,\mathcal{L}_{uv}]\mathbb{P}_u |\mathcal{O}) &= (\mathcal{O}|\mathbb{P}_u[F_u,\mathcal{L}_{uv}]\mathbb{P}_u \mathbb{P}_v |\mathcal{O}) + (\mathcal{O}|\mathbb{P}_v\mathbb{P}_u[F_u,\mathcal{L}_{uv}]\mathbb{P}_u (1-\mathbb{P}_v) |\mathcal{O}) \nonumber\\
    &= (\mathcal{O}|(2-\mathbb{P}_v)\mathbb{P}_u[F_u,\mathcal{L}_{uv}]\mathbb{P}_u\mathbb{P}_v |\mathcal{O}),
\end{align}
where we have used $(1-\mathbb{P}_v)\mathcal{L}_{uv}(1-\mathbb{P}_v)=0$, along with $F_u^\dagger = F_u$.   At this point, it is most helpful to separate out $b^\dagger_u b_v$ and $b^\dagger_v b_u$ terms in $H_{uv}$ and handle them separately.   Indeed, let us define \begin{subequations}
    \begin{align}
        \mathcal{L}^<_{u,v} |\mathcal{O}) &= \mathrm{i}J_{uv}(t) | b^\dagger_u b_v \mathcal{O}), \\
        \mathcal{L}^>_{u,v} |\mathcal{O}) &= -\mathrm{i}J_{uv}(t) |  \mathcal{O}b^\dagger_u b_v),
    \end{align}
\end{subequations}
so that we can split up \begin{equation}
    \mathcal{L}_{uv} \mathbb{P}_{\mathcal{B}_{uv}}|\mathcal{O}) = \left( \mathcal{L}_{uv}^< +\mathcal{L}_{uv}^> + \mathcal{L}_{vu}^< + \mathcal{L}_{vu}^>\right) \mathbb{P}_{\mathcal{B}_{uv}} |\mathcal{O}).
\end{equation}
As all terms are analyzed in exactly the same way, with the only differences being e.g. that \begin{subequations}
    \begin{align}
        \mathcal{L}^<_{uv}|\mathbf{n}) &= \mathrm{i}J_{uv}(t)|\mathbf{n} + \mathbf{a}_u -\mathbf{a}_v), \\
        \mathcal{L}^>_{uv}|\mathbf{n}) &=-\mathrm{i}J_{uv}(t)|\mathbf{n}+\mathbf{a}^\prime_v-\mathbf{a}^\prime_u)
    \end{align}
\end{subequations}, we will just focus on the first one $\mathcal{L}^<_{uv}$ in all cases which follow.  Since the interaction terms in the Hamiltonian obey \begin{equation}
    \mathcal{L}_U(t)|\mathbf{n}) = \mathrm{i}\frac{\mathrm{d}\theta_{\mathbf{n}}}{\mathrm{d}t}|\mathbf{n}),
\end{equation}
with $\mathrm{d}\theta_{\mathbf{n}}/\mathrm{d}t$ a conveniently named constant prefactor, we find that
\begin{align}
    [F_u,\mathcal{L}^<_{uv}]|\mathbf{n}) &= \mathrm{i}J_{uv}(t)\sqrt{(n_u+1)n_v} \mathrm{e}^{\mathrm{i}(\theta_{\mathbf{n} + \mathbf{a}_u -\mathbf{a}_v} - \theta_{\mathbf{n}})} \delta_{n_u\ge n_u^\prime} f(n_u)|\mathbf{n}+\mathbf{a}_u-\mathbf{a}_v)
\end{align}
 where \begin{align}
     f(n):&=(n+1+\beta)^\beta - (n+\beta)^\beta \notag \\
     &= (n+\beta)^{\beta-1} \sum^{\beta-1}_{k=0} \left(1+\frac{1}{n+\beta}\right)^k \le (n+\beta)^{\beta-1}\beta \left[\left(1+\frac{1}{\beta}\right)^\beta-1\right]\le (\mathrm{e}-1)\beta (n+\beta)^{\beta-1}.
\end{align}
Temporarily defining \begin{subequations}\begin{align}
    |(\mathbf{n}|\mathbb{P}_u(1-\mathbb{P}_v/2)|\mathcal{O})| &:= \varphi_{ \mathbf{n}}, \\
    |( \mathbf{n}|\mathbb{P}_u\mathbb{P}_v|\mathcal{O})| &:= \phi_{ \mathbf{n}},
\end{align}\end{subequations}
we see that  \begin{align}
   & \left||(\mathcal{O}|(2-\mathbb{P}_v)\mathbb{P}_u[F_u, \mathcal{L}^<_{uv}]\mathbb{P}_u\mathbb{P}_v|\mathcal{O})|\right| \le 2(\mathrm{e}-1)\beta\sum_{ \mathbf{n}}\phi_{\mathbf{n}}  (n_u+\beta)^{\beta-1}\sqrt{(n_u+1)n_v}\varphi_{\mathbf{n}+\mathbf{a}_u-\mathbf{a}_v} \notag \\
    &\qquad \le 2(\mathrm{e}-1)\beta\sum_{\mathbf{n}} \left[(n_v-1+\beta)^\beta + \delta_{\beta>1}(n_u+\beta)^\beta  \right] \phi_{\mathbf{n}}^2  +(n_u+\beta)^\beta \varphi_{\mathbf{n} + \mathbf{a}_u - \mathbf{a}_v}^2  \notag \\
    &\qquad  \le 2(\mathrm{e}-1)\beta \left[(1+\delta_{\beta>1})(\mathcal{O}|\mathcal{F}_u|\mathcal{O}) + (\mathcal{O}|\mathcal{F}_v|\mathcal{O})\right]. \label{eq:case1}
\end{align}
To obtain the second inequality above, we used: \begin{prop}\label{prop54}
Let $\xi_u,\xi_v,\varphi,\phi$ be positive real numbers, and $\beta$ be a positive integer.   Then \begin{align}
    \sqrt{\xi_u \xi_v}\xi_u^{\beta-1}\varphi\phi \le \xi_u^\beta \varphi^2 + \xi_v^\beta \phi^2 + \delta_{\beta>1} \xi_u^\beta \phi^2 \label{eq:xiphitrick}
\end{align}
\end{prop}
\begin{proof}This inequality is trivial for $\beta=1$ or $\varphi\phi=0$; the other cases can be proven by taking the ratio of the two sides of (\ref{eq:xiphitrick}), and using
\begin{align}
    (2\beta-1)\frac{\varphi^2+\phi^2}{\varphi\phi(2\beta-1)} \sqrt{\frac{\xi_u}{\xi_v}} + \frac{\phi}{\varphi} \left(\frac{\xi_v}{\xi_u}\right)^{\beta-1/2} \ge 2\beta \left[\left[\frac{\varphi^2+\phi^2}{\varphi\phi(2\beta-1)}\right]^{2\beta-1}\frac{\phi}{\varphi}  \right]^{1/2\beta} \ge (2\beta)^{\frac{1}{2\beta}}  \left[\left[\frac{\varphi^2+\phi^2}{\varphi\phi}\right]^{2\beta-1}\frac{\phi}{\varphi}  \right]^{1/2\beta}. \label{eq:prop54}
\end{align}
The first inequality comes from
\begin{equation}
(2\beta-1) a + b \ge 2\beta \left(a^{2\beta-1}b\right)^{1/2\beta},    
\end{equation}
with $(2\beta-1)a$ and $b$ the first two terms in the leftmost phrase of (\ref{eq:prop54}).  The second inequality in (\ref{eq:prop54}) comes from replacing $2\beta-1<2\beta$.  Now, letting $x=\phi/\varphi$, we observe that \begin{equation}
    \left[x \left(x+\frac{1}{x}\right)^{2\beta-1} \right]^{1/2\beta} = \left(x^2+1\right)^{1/2\beta} \left(x+\frac{1}{x}\right)^{1-1/\beta} \ge 1
\end{equation}
for any $x>0$. Hence we obtain (\ref{eq:xiphitrick}). 
\end{proof}

In the last line of (\ref{eq:case1}) we use Proposition \ref{propOFO} with $c=0$ and either $\mathbb{Q}=\mathbb{P}_v$ or $\mathbb{Q}= 1-\mathbb{P}_v/2$, both of which obey $\lVert \mathbb{Q}\rVert = 1$.  This completes Case 1.

\textbf{Case 2A:} The remaining 4 cases will all have a similar flavor.  The non-trivial aspect of these cases involves the presence of a $\mathbb{P}\mathcal{L}(1-\mathbb{P})$ term, which will require some special care: as in our proof of Proposition \ref{propOFO}, the $(1-\mathbb{P})$ projection onto the identity actually is  responsible for the fastest growing terms in our bound as $\mu \rightarrow 0$.   Assuming $z=u$, and defining \begin{subequations}\label{eq:case2phi}\begin{align}
|(\mathbf{n}|\mathbb{P}_u|\mathcal{O})| &:= \phi_{\mathbf{n}}, \\
    |(I_u\otimes \mathbf{n}_{-u}| [\mathbb{P}_v + \delta_{\ell>0}(1-\mathbb{P}_v)]|\overline{\mathcal{O}}_u)| &:= \psi_{\mathbf{n}_{-u}} , \end{align}
\end{subequations} we find that
\begin{align}
    &|(\tilde{\mathcal{O}}_u|F_u \mathcal{L}^<_{uv}|\overline{\mathcal{O}}_u)| \le  \sum_{\mathbf{n}}\delta_{n_un'_u} \phi_{\mathbf{n} + \mathbf{a}_u - \mathbf{a}_v} \psi_{\mathbf{n}_{-u}} \mathrm{e}^{-\mu n_u/2} \sqrt{1-\mathrm{e}^{-\mu}} \sqrt{(n_u+1)n_v}(n_u+1+\beta)^\beta \notag \\
    &\le \sum_{\mathbf{n}}\delta_{n_un'_u}\mathrm{e}^{-\mu n_u/2}(n_u+1+\beta) \left\{ \eta(1-\mathrm{e}^{-\mu}) \left[(n_v+\beta)^\beta+ \delta_{\beta>1}(n_u+1+\beta)^\beta\right] \psi_{\mathbf{n}_{-u}}^2 + \frac{1}{\eta}(n_u+1+\beta)^\beta \phi_{\mathbf{n} + \mathbf{a}_u - \mathbf{a}_v}^2\right\} \notag \\
    &\le 2\eta\left(\beta+\frac{1}{1-\mathrm{e}^{-\mu/2}}\right)\sum_{\mathbf{n}_{-u}}(n_v+\beta)^\beta\psi_{\mathbf{n}_{-u}}^2 + 2\eta\delta_{\beta>1}\left(\frac{\beta+1}{1-\ee^{-\mu/2}}\right)^{\beta+1} \sum_{\mathbf{n}_{-u} }\psi_{\mathbf{n}_{-u}}^2 + \frac{1}{\eta}\left(1+\beta+\frac{2}{\mathrm{e}\mu}\right)(\mathcal{O}|\mathcal{F}_u|\mathcal{O}). \label{eq:429}
\end{align}
In the first line, we have used (\ref{eq:Inorm}) to show that \begin{equation}
    \psi_{\mathbf{n}_{-u}} \sqrt{1-\mathrm{e}^{-\mu}} \mathrm{e}^{-\mu n_u/2} \ge |(\mathbf{n}|[\mathbb{P}_v + \delta_{\ell>0}(1-\mathbb{P}_v)]|\overline{\mathcal{O}}_u)|.
\end{equation}
In the second line, we introduced an arbitrary new constant $0<\eta<\infty$, by noting that \begin{equation}
    \sqrt{1-\mathrm{e}^{-\mu}} \phi \psi = \left(\sqrt{1-\mathrm{e}^{-\mu}}  \psi \sqrt{\eta} \right) \times \frac{\phi}{\sqrt{\eta}}.
\end{equation}
and using Proposition \ref{prop54}.  In the third line, we used (\ref{eq:sum_na}) to explicitly evaluate $n_u$ sums in the first two terms, along with the inequality
\begin{equation}\label{eq:nae}
    n^a \mathrm{e}^{-b n} < \left(\frac{a}{\mathrm{e}b}\right)^a, \;\;\; (\text{for all }0\le n<\infty),
\end{equation}
 in order to efficiently handle the extra factor of $\mathrm{e}^{-\mu n_u/2}(n_u+1+\beta)$ in the third term.   
 
 For the second term in the last line of (\ref{eq:429}), we can easily see that (recall (\ref{eq:PBuv})) \begin{equation}
     \sum_{\mathbf{n}_{-u}} \psi_{\mathbf{n}_{-u}}^2 \le (\mathcal{O}|\mathbb{P}_{\mathcal{B}_{uv}}|\mathcal{O}).
 \end{equation}
 To simplify the first term in (\ref{eq:429}), we use Proposition \ref{propOFO} with $\mathbb{Q}_{-v} = 1-\mathbb{P}_u$ and $c= \delta_{\ell>0}$: \begin{equation}
     \sum_{\mathbf{n}_{-u}}(n_v+\beta)^\beta\psi_{\mathbf{n}_{-u}}^2 \le 2 (\mathcal{O}|\mathcal{F}_v|\mathcal{O}) + 2 \delta_{\ell >0} \left(\frac{\beta}{1-\mathrm{e}^{-\mu}}\right)^{\beta} (\mathcal{O}|\mathbb{P}_{\mathcal{B}_{uv}}|\mathcal{O}). \label{eq:549}
 \end{equation}
 Now using $\eta = 1/2$ in (\ref{eq:429}), we conclude the analysis of Case 2A:
\begin{align}
|(\tilde{\mathcal{O}}_u|F_u \mathcal{L}^<_{uv}|\overline{\mathcal{O}}_u)| 
&\le 2\left(1+\beta+\frac{2}{\mu}\right)[(\mathcal{O}|\mathcal{F}_u|\mathcal{O})+(\mathcal{O}|\mathcal{F}_v|\mathcal{O})] + (\delta_{\beta>1}+ 2\delta_{\ell>0}) \left(\frac{\beta+1}{1-\ee^{-\mu/2}}\right)^{\beta+1} (\mathcal{O}|\mathbb{P}_{\mathcal{B}_{uv}}|\mathcal{O}). \label{eq:case2A}
\end{align}

\textbf{Case 2B:} Now we turn to the case $z\ne u,v$, which contributes only when $\ell>0$.  Now defining \begin{subequations}\label{eq:case2phi}\begin{align}
|(\mathbf{n}|\mathbb{P}_z|\mathcal{O})| &:= \phi_{\mathbf{n}}, \\
    |(I_u\otimes \mathbf{n}_{-z}|\overline{\mathcal{O}}_z)| &:= \psi_{\mathbf{n}_{-z}} , \end{align}
\end{subequations}
we find that
\begin{align}
    |(\tilde{\mathcal{O}}_z|F_z \mathcal{L}^<_{uv}|\overline{\mathcal{O}}_z)| &\le \delta_{\ell>0}\sum_{\mathbf{n}} \delta_{n_zn_z^\prime}(n_z+\beta)^\beta \phi_{\mathbf{n}+\mathbf{a}_u - \mathbf{a}_v} \sqrt{(n_u+1)n_v}\sqrt{1-\mathrm{e}^{-\mu}}\mathrm{e}^{-\mu n_z/2} \psi_{\mathbf{n}_{-z}} \notag \\ &\le \frac{\delta_{\ell>0}}{2} \sum_{\mathbf{n}}\delta_{n_zn_z^\prime}\mathrm{e}^{-\mu n_z/2} (n_z+\beta)^\beta  \left[(n_u+1)\phi_{\mathbf{n}+\mathbf{a}_u - \mathbf{a}_v}^2 + (1-\mathrm{e}^{-\mu})  n_v\psi_{\mathbf{n}_{-z}}^2 \right] \notag \\
    &\le \frac{\delta_{\ell>0}}{2}  \left[\sum_{\mathbf{n}}(2\beta)^\beta \left(1+\left(\frac{2}{\mathrm{e}\mu}\right)^{2\beta}\right)(n_u+1)\phi_{\mathbf{n}+\mathbf{a}_u - \mathbf{a}_v}^2 +  \sum_{\mathbf{n}_{-z}} \frac{(1-\mathrm{e}^{-\mu})\beta^\beta}{(1-\mathrm{e}^{-\mu/2})^{\beta+1}} n_v\psi_{\mathbf{n}_{-z}}^2 \right].
\end{align}
In the second line we used $ab \le \frac{1}{2}(a^2+b^2)$; in the third line, we used (\ref{eq:sum_na}) together with \begin{align}
    (n_z+\beta)^\beta \mathrm{e}^{-\mu n_z/2} &\le 2^\beta \left(\beta^\beta + n_z^\beta\right)\mathrm{e}^{-\mu n_z/2} \le (2\beta)^\beta \left(1+\left(\frac{2}{\mathrm{e}\mu}\right)^{\beta}\right);
\end{align} 
the last inequality follows from (\ref{eq:nae}).  Lastly, we use that (e.g.) $n_u+1 \le (n_u+\beta)^\beta$ along with analogous manipulations to (\ref{eq:549}) to see that \begin{align}
    |(\tilde{\mathcal{O}}_z|F_z \mathcal{L}^<_{uv}|\overline{\mathcal{O}}_z)| &\le \delta_{\ell>0} (2\beta)^\beta \left(1+\left(\frac{2}{\mathrm{e}\mu}\right)^{\beta}\right) \left[ (\mathcal{O}|\mathcal{F}_u|\mathcal{O}) + \left(\frac{\beta}{1-\mathrm{e}^{-\mu}}\right)^{\beta} (\mathcal{O}|\mathbb{P}_{\mathcal{B}_{uv}}|\mathcal{O})\right] \notag \\
    &\qquad + \delta_{\ell>0} \frac{2\beta^\beta}{(1-\mathrm{e}^{-\mu/2})^{\beta}}\left[ (\mathcal{O}|\mathcal{F}_v|\mathcal{O}) + \left(\frac{\beta}{1-\mathrm{e}^{-\mu}}\right)^{\beta} (\mathcal{O}|\mathbb{P}_{\mathcal{B}_{uv}}|\mathcal{O})\right]. \label{eq:case2B}
\end{align}

\textbf{Case 3A:} Let $z=u$.  Now denote \begin{subequations}\begin{align}
|(\mathbf{n}|\mathbb{P}_u|\mathcal{O})| &:= \phi_{ \mathbf{n}}, \\
    |(I_u\otimes \mathbf{n}_{-u}| [\mathbb{P}_v + \delta_{\ell>0}(1-\mathbb{P}_v)] F_u|\tilde{\mathcal{O}}_u)| &:= \psi_{\mathbf{n}_{-u}} . \end{align}
\end{subequations}
Since $(1-\mathbb{P}_u)\mathcal{L}_{uv} = (1-\mathbb{P}_u)\mathcal{L}_{uv}\mathbb{P}_u$ ($\mathcal{L}_{uv}$ will always change either $n_u$ or $n_u^\prime$), we may simply evaluate
\begin{align}
    |(\tilde{\mathcal{O}}_u|F_u (1-\mathbb{P}_u) \mathcal{L}^<_{uv}|\tilde{\mathcal{O}}_u)| &\le \sum_{\mathbf{n}} \delta_{n_un_u^\prime} \sqrt{1-\mathrm{e}^{-\mu}}\mathrm{e}^{-\mu n_u/2}\psi_{\mathbf{n}_{-u}} \sqrt{n_u(n_v+1)} \phi_{\mathbf{n}+\mathbf{a}_v-\mathbf{a}_u} \notag\\ 
    &\le \sum_{\mathbf{n}}\delta_{n_un_u^\prime}\mathrm{e}^{-\mu n_u/2} \left[\frac{\eta}{4}(1-\mathrm{e}^{-\mu})(n_v+1)\psi_{\mathbf{n}_{-u}}^2 + \frac{1}{\eta}(n_u+1) \phi_{\mathbf{n}+\mathbf{a}_v-\mathbf{a}_u}^2 \right] \notag\\ 
    &\le \frac{\eta}{2} \sum_{\mathbf{n}_{-\mathbf{u}}}(n_v+1)\psi_{\mathbf{n}_{-u}}^2 + \frac{1}{\eta} (\mathcal{O}|\mathcal{F}_u|\mathcal{O}), \label{eq:554}
    \end{align}
    employing similar tricks to Case 2B.   For the first term, define $\mathbb{Q} = (1-\mathbb{P}_u) F_u$, and observe that \begin{align}
    \lVert\mathbb{Q}\rVert = \lVert(1-\mathbb{P}_u) F_u\rVert = \lVert F_u|I_u)\rVert_2 = \sqrt{1-\mathrm{e}^{-\mu}} \sqrt{\sum_n (n+\beta)^{2\beta}\mathrm{e}^{-\mu n} } \le \left(\frac{2\beta}{1-\mathrm{e}^{-\mu}}\right)^\beta. \label{eq:555}
\end{align}
Similarly to Proposition \ref{propOFO}, 
\begin{align}
   \sum_{\mathbf{n}_{-uv}}\psi_{\mathbf{n}_{-u}}^2 &= \sum_{\mathbf{n}_{-uv}}|(I_u\otimes \mathbf{n}_{-u}| (1-\mathbb{P}_u)F_u [\mathbb{P}_v + \delta_{\ell>0}(1-\mathbb{P}_v)] |\tilde{\mathcal{O}}_u)|^2 \notag \\
   &\le  \lVert (1-\mathbb{P}_u)F_u \rVert^2 \times (I|I) \times \lVert \mathbb{P}_v^{nn^\prime}[\mathbb{P}_v + \delta_{\ell>0}(1-\mathbb{P}_v)] |\tilde{\mathcal{O}}_u)\rVert_2^2. \label{eq:556}
\end{align}
    Plugging (\ref{eq:555}) and (\ref{eq:556}) into (\ref{eq:554}), noting that $(I|I)=1$,  and using Proposition \ref{propOFO}, we find \begin{align}
      \sum_{\mathbf{n}_{-\mathbf{u}}}(n_v+1)\psi_{\mathbf{n}_{-u}}^2 \le \eta \left(\frac{2\beta}{1-\mathrm{e}^{-\mu}}\right)^{2\beta}  \left[ (\mathcal{O}|\mathcal{F}_v|\mathcal{O}) + \delta_{\ell>0}\left(\frac{\beta}{1-\mathrm{e}^{-\mu}}\right)^{\beta} (\mathcal{O}|\mathbb{P}_{\mathcal{B}_{uv}}|\mathcal{O})\right] + \frac{(\mathcal{O}|\mathcal{F}_u|\mathcal{O})}{\eta}.
    \end{align}
    Choosing \begin{equation}
        \eta = \left(\frac{2\beta}{1-\mathrm{e}^{-\mu}}\right)^{-\beta}
    \end{equation}
    we obtain \begin{equation}
        |(\tilde{\mathcal{O}}_u|F_u (1-\mathbb{P}_u) \mathcal{L}^<_{uv}|\tilde{\mathcal{O}}_u)| \le \left(\frac{2\beta}{1-\mathrm{e}^{-\mu}}\right)^{\beta} \left[ (\mathcal{O}|\mathcal{F}_u|\mathcal{O})+ (\mathcal{O}|\mathcal{F}_v|\mathcal{O}) + \delta_{\ell>0}\left(\frac{\beta}{1-\mathrm{e}^{-\mu}}\right)^{\beta} (\mathcal{O}|\mathbb{P}_{\mathcal{B}_{uv}}|\mathcal{O})\right]. \label{eq:case3A}
    \end{equation}

\textbf{Case 3B:} The last case proceeds very similarly to Case 3A.  Defining \begin{subequations}\begin{align}
|(\mathbf{n}|\mathcal{O})| &:= \phi_{ \mathbf{n}}, \\
    |(I_z\otimes \mathbf{n}_{-z}|  F_z|\tilde{\mathcal{O}}_z)| &:= \psi_{\mathbf{n}_{-z}}, \notag \\
    |( \mathbf{n}_{-z}|  (1-\mathbb{P}_z)|\mathcal{O})| &:= \tilde\psi_{\mathbf{n}_{-z}} \end{align}
\end{subequations}
and noting that analogous to (\ref{eq:556}), \begin{equation}
    \sum_{\mathbf{n}_{-zv}}\psi_{\mathbf{n}_{-z}}\le \left(\frac{2\beta}{1-\mathrm{e}^{-\mu}}\right)^{\beta} \sum_{\mathbf{n}_{-zv}} \tilde\psi_{\mathbf{n}_{-z}} \label{eq:561}
\end{equation}
we find that
\begin{align}
    |(\tilde{\mathcal{O}}_z|F_z (1-\mathbb{P}_z) \mathcal{L}_{uv}^<|\mathcal{O})| &\le \delta_{\ell>0} \sum_{\mathbf{n}} \sqrt{1-\mathrm{e}^{-\mu}}\mathrm{e}^{-\mu n_z/2}\psi_{\mathbf{n}_{-z}} \sqrt{(n_v+1)n_u} \phi_{\mathbf{n}+\mathbf{a}_v-\mathbf{a}_u} \delta_{n_zn_z^\prime} \notag\\ &\le \frac{\delta_{\ell>0}}{2} \left(\frac{2\beta}{1-\mathrm{e}^{-\mu}}\right)^{\beta} \sum_{\mathbf{n}} \delta_{n_zn_z^\prime} \left[(1-\mathrm{e}^{-\mu}) \mathrm{e}^{-\mu n_z}(n_v+1)\tilde\psi_{\mathbf{n}_{-z}}^2 + n_u \phi_{\mathbf{n}+\mathbf{a}_v-\mathbf{a}_u}^2 \right] \notag\\ 
    &\le \frac{\delta_{\ell >0}}{2}\left(\frac{2\beta}{1-\mathrm{e}^{-\mu}}\right)^{\beta} \left[ \sum_{\mathbf{n}_{-z}}(n_v+1)\tilde\psi_{\mathbf{n}_{-z}}^2+ \sum_{\mathbf{n}}n_u \phi_{\mathbf{n}+\mathbf{a}_v-\mathbf{a}_u}^2  \right] \notag \\
    &\le \delta_{\ell >0}\left(\frac{2\beta}{1-\mathrm{e}^{-\mu}}\right)^{\beta}\left[ (\mathcal{O}|\mathcal{F}_u|\mathcal{O}) + (\mathcal{O}|\mathcal{F}_v|\mathcal{O}) + 2\left(\frac{\beta}{1-\mathrm{e}^{-\mu}}\right)^{\beta} (\mathcal{O}|\mathbb{P}_{\mathcal{B}_{uv}}|\mathcal{O}) \right] \label{eq:case3B}
\end{align}
where we completed the square in the second line along with using (\ref{eq:561}), evaluated the sum over $n_z$ in the third line, and used Proposition \ref{propOFO} in the fourth line.

\textbf{Combining the cases:} Now it simply remains to combine all of our results: (\ref{eq:case1}) for Case 1, (\ref{eq:case2A}) for Case 2A, (\ref{eq:case2B}) for Case 2B, (\ref{eq:case3A}) for Case 3A, and (\ref{eq:case3B}) for Case 3B.  We will use many elementary inequalities to try and simplify complicated expressions, such as \begin{equation}
    \frac{1}{1-\mathrm{e}^{-\mu/2}} \le 1 + \frac{2}{\mu},
\end{equation}$\beta+1<2\beta$, etc., along with the (quite loose) inequality (\ref{eq:PPF}).   When $\ell=0$, we may simply replace $(\mathcal{O}|\mathbb{P}_{\mathcal{B}_{uv}}|\mathcal{O}) \le (\mathcal{O}|\mathcal{F}_u|\mathcal{O}) + (\mathcal{O}|\mathcal{F}_v|\mathcal{O})$.   We then observe that in the above calculation, it is this combination of $(\mathcal{O}|\mathcal{F}_u|\mathcal{O}) + (\mathcal{O}|\mathcal{F}_v|\mathcal{O})$ which shows up everywhere.   This then implies that our bound on $M_{uu}$ will be $K$ times larger than our bound on $M_{uv}$, where we have used that (as defined above) no vertex in $G$ has more than $K$ adjacent vertices.   This leads us to the $\ell=0$ cases contained in  (\ref{eq:Muv}).

For simplicity, we get a little bit lazier with the $\ell>0$ cases.   Firstly, let us simply use the crude fact above that \begin{equation}
    (\mathcal{O}|\mathcal{F}_u|\mathcal{O}),  (\mathcal{O}|\mathcal{F}_v|\mathcal{O}),  (\mathcal{O}|\mathbb{P}_{\mathcal{B}_{uv}}|\mathcal{O}) \le \sum_{x\in\mathcal{B}_{uv}} (\mathcal{O}|\mathcal{F}_x|\mathcal{O}).  \label{eq:lazyboundsell}
\end{equation}
It is then simply a matter of counting up every single coefficient.   Observe that for a given edge $(uv)\in E$, we may induce a contribution to $M_{xy}(t)$ for $x,y\ne u,v$.  The following proposition bounds how often this can happen: \begin{prop}Consider two vertices $\lbrace x,y\rbrace \subset V$ in a graph $G=(V,E)$ with maximal degree $K$.  Recall the subsets $\mathcal{B}_e$, defined in (\ref{eq:defBuv}) for each edge $e\in E$.  Let the number of edges $e$ for which $\lbrace x,y\rbrace \subseteq \mathcal{B}_{uv}$ be defined as $\mathcal{N}_{xy}$.  Then \begin{equation}
    \mathcal{N}_{xy} := |\lbrace e\in E : \lbrace x,y\rbrace \subseteq \mathcal{B}_{e}\rbrace| \le \delta_{\mathrm{dist}(x,y)\le 2\ell +1} K^{\ell+1}. \label{eq:prop56}
\end{equation} \label{prop56}
\end{prop}   
\begin{proof}
If $\ell=0$, then $\mathcal{N}_{xy}=1$: $e=(xy)$ is required.  So (\ref{eq:prop56}) is true but loose, in this case.

If $\ell>0$, observe that we can (lazily) bound $\mathcal{N}_{xy}$ by simply finding the number of $\mathcal{B}_e$ containing $x$.  This is upper bounded by assuming that the graph $G$ is a $K$-regular tree:  the reason is because if $G$ contains any cycles (loops), then it is possible that the following count (based on the assumption of a tree) of the number of edges $e$ within a distance $\ell$ of $x$ may double count edges.   On a $K$-regular tree, there are $K$ neighbors $u$ of the vertex $x$.  Each $u$ has $K-1$ additional neighbors $u^\prime$, with $\mathrm{dist}(u^\prime,x)=2$.  Continuing this process, we see that there are $K(K-1)^m$ edges that connected a vertex a distance $m$ from $x$ to a vertex at distance $m+1$.   Then \begin{equation}
    \mathcal{N}_{xy} \le \sum_{m=0}^\ell K(K-1)^{m} \le K+ \sum_{m=1}^{\ell } K^m(K-1) = K^{\ell+1}, \label{eq:Nxybound}
\end{equation}
which completes the proof.
\end{proof}

Proposition \ref{prop56} implies that for any pair of vertices $x,y$, we may have contributions to $M_{xy}$ from up to $\mathcal{N}_{xy}$ couplings in $\mathcal{L}_J$.   So, summing up the total contribution from a single coupling using (\ref{eq:lazyboundsell}), we arrive at the $\ell>0$ results in (\ref{eq:Muv}).
\end{proof}

The hard part of the proof is now complete.  The last step is rather standard: to solve the differential equations (\ref{eq:dCudt}) and bound the resulting $C_v(t)$.  We achieve this using ``quantum walk inspired" methods, following \cite{Yin:2020pjd}:
\begin{lem}\label{lem56}
    Given a graph $G=(V,E)$ and real-valued functions $C_v(t)$ on each vertex $v$, if the differential inequalities
    \begin{equation}
        \frac{\mathrm{d}C_v}{\mathrm{d}t} \le A_v(t) C_v(t) + \sum_{u: \mathrm{dist}(u,v)\le 2\ell+1} B_{uv}(t) C_u(t),
    \end{equation}
    then if \begin{subequations} \label{eq:AvBuvBounds}
        \begin{align}
            A_v(t) &\le K^{2\ell+1}B, \\
            B_{uv}(t) &\le B,
        \end{align}\end{subequations}
        and the initial conditions are that (for subset $R\subset V$) $C_v(0)=0$ if $v\notin R$, then if $\mathrm{dist}(x,R) = r$,
        \begin{equation}
            C_x(t) \le \left(\frac{vt}{r}\right)^{r/(2\ell+1)} \times \sum_{x\in R} C_x(0), \; \;\; \text{ if } vt<r, \label{eq:lemma56}
        \end{equation}
        where the velocity \begin{equation}
            v < 4(2\ell+1)K^{2\ell+1}B. \label{eq:lemma56v}
        \end{equation}
\end{lem}
\begin{proof}
Let $\lambda>1$ be a real number, and define \begin{equation}
    G(t) := \sum_{v\in V} C_v(t) \lambda^{\mathrm{dist}(v,R)}.
\end{equation}
Observe that, using (\ref{eq:AvBuvBounds}), \begin{equation}
    \frac{\mathrm{d}G}{\mathrm{d}t} \le \sum_{v\in V} \left[ K^{2\ell+1}B C_v(t) + \sum_{u: \mathrm{dist}(u,v)\le 2\ell+1} BC_u(t) \right] \lambda^{\mathrm{dist}(v,R)} \le K^{2\ell+1}B \left(1+\lambda^{2\ell+1}\right)G(t),
\end{equation}
where in the second equality we used that $\lambda^{\mathrm{dist}(v,R)}\le \lambda^{2\ell+1+\mathrm{dist}(u,R)}$, along with the fact that the number of vertices $u$ within distance $2\ell+1$ of any given vertex must be $\le K^{2\ell+1}$, analogously to (\ref{eq:Nxybound}).  Therefore, \begin{equation}
    G(t) \le G(0) \exp\left[K^{2\ell+1}B\left(1+\lambda^{2\ell+1}\right)t\right].
\end{equation}
In the spirit of Markov's inequality, we thus find that if $r=\mathrm{dist}(x,R)$, \begin{equation}
    C_x(t) \le \lambda^{-r}G(t) \le G(0) \exp\left[K^{2\ell+1}B\left(1+\lambda^{2\ell+1}\right)t - \frac{r}{2\ell+1}\log \lambda^{2\ell+1} \right].
\end{equation}
We now choose the optimal value of $\lambda$, which corresponds to \begin{equation}
    \lambda^{2\ell+1} = \frac{r}{(2\ell+1)K^{2\ell+1}Bt}.
\end{equation}
We then find that \begin{equation}
    C_x(t) \le G(0)\exp\left[-\frac{r}{2\ell+1} \left(\log\frac{r}{(2\ell+1)K^{2\ell+1}Bt} - \frac{(2\ell+1)K^{2\ell+1}Bt}{r}-1 \right)\right]. \label{eq:56C}
\end{equation}
If the object in parentheses above is positive, then $C_x(t)$ is super-exponentially suppressed.  It is straightforward to numerically check that \begin{equation}
    \log x - \frac{1}{x}-1 > \log \frac{x}{4} > 0, \;\;\; (4<x<\infty). \label{eq:56x}
\end{equation}
Combining (\ref{eq:56C}) and (\ref{eq:56x}), and using that \begin{equation}
    G(0) = \sum_{x\in R} C_x(0),
\end{equation} 
we find that (\ref{eq:lemma56}) holds for velocity $v$ given in (\ref{eq:lemma56v}).
\end{proof}

According to Lemma \ref{lem52}, $G(0)$ in the previous proof is
\begin{equation}
    G(0)= 2\sum_{x\in R} (\mathcal{O}|F_x|\mathcal{O}) + 2\left(\frac{\beta}{1-\mathrm{e}^{-\mu}}\right)^{\beta} (|R|+|R_\ell|),
\end{equation}
where $R_\ell=\{x\in V: \mathrm{dist}(x,R)\le \ell\}$. Then
(\ref{eq:thm1}) and (\ref{eq:thm2}) immediately follow from combining (\ref{eq:dCudt}) with Proposition \ref{prop33}, and Lemmas \ref{lem53} and \ref{lem56}.   We have thus proven the existence of a linear light cone in the grand canonical ensemble of interacting bosonic models.
\end{proof}

Note that in the case $\ell>0$, we actually know that $A_v(t)\le B$ as well, and so the bound in (\ref{eq:thm2}) is expected to be particularly weak in this case -- however, as noted in the introduction, we believe that none of our O(1) coefficients are particularly tight; the most important result in this theorem (besides the fact $v$ is finite!) is the scaling of velocity when $\beta=0$ and $\ell=1$, which cannot qualitatively be improved any further.  

On a nearest neighbor $d$-dimensional cubic lattice, one has $K=2d$ and thus in higher dimensions our velocity factor becomes larger. This effect is common to Lieb-Robinson bounds \cite{PRXQuantum.1.010303}, and arises in such a cubic lattice due to the fact that there are exponentially many paths one can find between two widely separated points.  There is a contribution to our commutator bound and quantum walk from operators growing along each path.

\section{One dimensional models}\label{sec:1d}
One important limitation of Theorem \ref{thm51} is that it only holds for ``thermal averages" in a particular infinite temperature grand canonical ensemble.  While such a result is highly suggestive that a light cone exists in \emph{all} finite density states, it does not represent a mathematically rigorous proof.  In this section, we will show that in one dimensional models, we can come very close to proving a ``worst case" Lieb-Robinson-style bound, which demonstrates a finite velocity of quantum information in \emph{all} finite density states. Furthermore, we can remove the $\beta$ dependence of the information speed, so that all physical processes are bounded by one speed, regardless of what operator to probe the system.

In order to do this, we first introduce some notation.  Let $V=\{i: i=-L,-L+1,\cdots, L\}$ denote sites in a 1d chain, labeled by integers.  Define $\mathbb{Q}_x$ $(x\ge 0)$ to project onto operators acting nontrivially on the set $\{x,-x\}$ but no further measured from $i=0$:
\begin{align}
    \mathbb{Q}_x = \mathbb{P}_{\{x,-x\}} \prod_{y>x}(1-\mathbb{P}_y)(1-\mathbb{P}_{-y}).
\end{align}
Immediately we notice the following useful result:
\begin{prop}
If $\mathbb{Q}_0|\mathcal{O}) = |\mathcal{O})$, then  \begin{align}\label{eq:Q<P<}
    (\mathcal{O}(t) | \mathbb{Q}_r|\mathcal{O}(t) ) \le (\mathcal{O}(t) | \mathbb{P}_{\{r,-r\}}|\mathcal{O}(t) ) \le C \left(\frac{vt}{r}\right)^{r/(2l+1)},
\end{align}
with \begin{equation}
    v < \left\lbrace\begin{array}{ll} 8K(31+24\mu^{-1}) &\ \ell = 0 \\ 2^{11}(2l+1)K^{3\ell+2}(1+2\mu^{-1})^2 &\ \ell>0  \end{array}\right.. \label{eq:6v}
\end{equation}
\end{prop}
\begin{proof}
Since $\norm{1-\mathbb{P}_j} = 1$, we see that $(\mathcal{O}(t) | \mathbb{Q}_r|\mathcal{O}(t) ) \le (\mathcal{O}(t) | \mathbb{P}_{\{r,-r\}}|\mathcal{O}(t) )$.  To bound this latter inner product, we use Lemma \ref{lem56}.  This shows us that (\ref{eq:Q<P<}) holds; moreover, $v$ can be evaluated at $\beta=1$, which leads to (\ref{eq:6v}).
\end{proof}

Using this proposition, we can then prove the following theorem:
\begin{thm}\label{thm71}
Let $R= \{i\in V: r\le i\le r_+\}$, where $r_+-r=O(1)$. Define $\mathcal{O}^\prime, \beta, \gamma$ as in (\ref{eq:O'}) and (\ref{eq:beta_gamma}).
If there are some $\mu, \theta,K_0>0$ such that the state $\tilde{\rho}$ satisfies
\begin{equation}
    \mathrm{tr} \left(\sqrt{ \tilde{\rho} } A^\dagger \sqrt{ \tilde{\rho} } A \right) \le K_0\theta^{2x} \mathrm{tr}\left(\sqrt{\rho_\mu} A^\dagger \sqrt{\rho_\mu} A\right), \quad \forall A = A_{\le x}\otimes I_{>x}, \label{eq:Aansatz}
\end{equation}
(i.e. $A$ is non-identity only within sites $\{-x,\cdots,x\}$),
then we have the inequality
\begin{equation}
  ([\mathcal{O}(t),\mathcal{O}^\prime] | [\mathcal{O}(t),\mathcal{O}^\prime])_{\tilde{\rho}}:=\mathrm{tr}\left(\sqrt{\tilde{\rho}}[\mathcal{O}(t),\mathcal{O}^\prime]^\dagger \sqrt{\tilde{\rho}} [\mathcal{O}(t),\mathcal{O}^\prime] \right) \le C_1 \left(\frac{(2\theta)^{8l+4}v^\prime t}{r}\right)^{r/(2l+1)}, \label{eq:7thm}
\end{equation}
for $r>(2\theta)^{8l+4}v^\prime t$. Here $v^\prime = (1+\epsilon) v_{\mu/2}$ where $v_\mu$ is given in (\ref{eq:6v}) and $\epsilon$ is arbitrarily small but finite.  The constants $0<C_1,\epsilon <\infty$ are independent of $r$.
\end{thm}

\begin{proof}
We will prove this result in 2 steps: first, we will analyze inner products of the form $(\mathcal{O}\mathcal{O}^\prime|\mathcal{OO}^\prime)$ without relying on an $F$-ansatz (as we did the previous sections);  then, we will show how to use (\ref{eq:Aansatz}) in order to obtain (\ref{eq:7thm}).

Let us begin with our first step.  In what follows, we denote $\mathcal{O}(t)$ by $\mathcal{O}$. Since obviously $[\mathcal{O}(t),\mathcal{O}^\prime]=0$ if $\mathcal{O}(t)$ has no support in the set $R$, we can always project $\mathcal{O}(t)$ onto operators that have support in set $R$.   It turns out to be convenient to do this using the $\mathbb{P}_R$ operator introduced above -- but with an inner product evaluated at $\mu/2$ instead!  (We will point out later where this ``trick" becomes useful.)   Using the Cauchy-Schwarz inequality, we find that  \begin{equation} \label{eq:triang_ine}
    ([\mathcal{O},\mathcal{O}^\prime] | [\mathcal{O},\mathcal{O}^\prime])_{\tilde{\rho}} \le 2 \left( \mathcal{O}^\prime (\mathbb{P}_R^{\mu/2} \mathcal{O}) |  \mathcal{O}^\prime (\mathbb{P}_R^{\mu/2} \mathcal{O}) \right)_{\tilde{\rho}} + 2 \left( (\mathbb{P}_R^{\mu/2} \mathcal{O})\mathcal{O}^\prime  |  (\mathbb{P}_R^{\mu/2} \mathcal{O})\mathcal{O}^\prime  \right)_{\tilde{\rho}}
\end{equation}
where $\mathbb{P}^{\mu/2}_R$ is the projection operator defined via the inner product $\rho_{\mu/2}$.  In the rest of this proof we neglect to write the superscript $\mu/2$ in $\mathbb{P}_R$.

It is useful to expand out $\mathbb{P}_R\mathcal{O}$ a bit more explicitly.  We write \begin{equation}\label{eq:Ox}
   \mathbb{P}_R \mathcal{O} = \mathcal{O}_{\le r_+} + \sum_{x=r_++1}^L \mathcal{O}_x,
   \end{equation}
where \begin{subequations}\label{eq:68ab}\begin{align}
    \mathcal{O}_x &:= \mathbb{Q}_x\mathbb{P}_R\mathcal{O} = \sum_{\mathbf{n},\mathbf{n}^\prime}\tilde{\mathcal{O}}_{x,\mathbf{n}\mathbf{n}^\prime}|\mathbf{n}\rangle \langle \mathbf{n}^\prime| \otimes I_{>x} , \\
    \mathcal{O}_{\le x} &:= \mathbb{P}_R\mathcal{O}- \sum_{y=x+1}^L \mathcal{O}_x = \sum_{\mathbf{n}, \mathbf{n}^\prime }\tilde{\mathcal{O}}_{\le x,\mathbf{n}\mathbf{n}^\prime}|\mathbf{n}\rangle \langle \mathbf{n}^\prime| \otimes I_{>x}
\end{align}\end{subequations}
where $\mathbf{n},\mathbf{n}^\prime$ above only run over sites $\{-x,\cdots, x\}$, $\mathbf{n}=\{n_{-x}n_{-x}^\prime,\cdots,n_xn_x^\prime\}$, and $|\mathbf{n}\rangle\langle \mathbf{n}^\prime|$ is shorthand for $|n_{-x}\cdots n_x\rangle\langle n_{-x}^\prime\cdots n_x^\prime|$.  Here we are temporarily using the ``bare'' operator basis $|n\rangle\langle n'|$, whose coefficient is $\tilde{\mathcal{O}}_{\mathbf{n}}$ (this is \emph{not} the same as our previously introduced $\mathcal{O}_{\mathbf{n}}$). Observe that from (\ref{eq:Q<P<}),
\begin{align}
    \mathrm{tr}\left(\sqrt{\rho_{\mu/2} } \mathcal{O}_x^\dagger \sqrt{\rho_{\mu/2}} \mathcal{O}_x  \right) &= \sum_{\mathbf{n} \in \mathbf{n}_{\le x} } |\tilde{\mathcal{O}}_{x,\mathbf{n}}|^2 \prod_{|i|\le x} (1-\mathrm{e}^{-\mu/2}) \mathrm{e}^{-\mu (n_i+n_i^\prime)/4} \le C \left(\frac{v t}{x}\right)^{x/(2l+1)}, \\
    \mathrm{tr}\left(\sqrt{\rho_{\mu/2} } \mathcal{O}_{\le x}^\dagger \sqrt{\rho_{\mu/2}} \mathcal{O}_{\le x}  \right) &= \sum_{\mathbf{n} \in \mathbf{n}_{\le x} } |\tilde{\mathcal{O}}_{\le x,\mathbf{n}}|^2 \prod_{|i|\le x} (1-\mathrm{e}^{-\mu/2}) \mathrm{e}^{-\mu (n_i+n_i^\prime)/4} \le (\mathcal{O}(t) | \mathbb{P}_R|\mathcal{O}(t) ) \le  C^\prime \left(\frac{vt}{r}\right)^{r/(2l+1)}.
\end{align}

Now, let us analyze what multiplication by $\mathcal{O}^\prime$ does.  Similar to our discussion in the proof of Proposition \ref{prop33} (and using similar notation), we observe that  \begin{align}
    \mathcal{O}^\prime \mathcal{O}_x = \sum_{\mathbf{n} \in \mathbf{n}_{\le x}}\tilde{\mathcal{O}}_{x,\mathbf{n}} c_{\mathbf{n}} |\mathbf{n} + \mathbf{g}\rangle\langle \mathbf{n}| \otimes I_{>x},
\end{align}
where, using (\ref{eq:413}), \begin{equation}
    0\le c_{\mathbf{n}} \le \left(\beta + \sum_{x\in R} n_x\right)^{\beta/2}.  \label{eq:cnbound}
\end{equation}
An analogous calculation to what follows holds for $\mathcal{O}_x\mathcal{O}^\prime$, as well as for $\mathcal{O}^\prime \mathcal{O}_{\le x}$, so we will show only the case $\mathcal{O}^\prime \mathcal{O}_x$ explicitly.
Using the inner product induced by $\rho_\mu$,
\begin{align}
    \left(\mathcal{O}^\prime \mathcal{O}_x| \mathcal{O}^\prime \mathcal{O}_x \right)_\mu &= \sum_{\mathbf{n} \in \mathbf{n}_{\le x}} |\tilde{\mathcal{O}}_{x,\mathbf{n}}|^2 c_{\mathbf{n}}^2 \langle \mathbf{n}|\sqrt{\rho} | \mathbf{n}\rangle \langle \mathbf{n}+\mathbf{g}|\sqrt{\rho} | \mathbf{n}+\mathbf{g}\rangle  \notag \\
    &= \sum_{\mathbf{n} \in \mathbf{n}_{\le x}} |\tilde{\mathcal{O}}_{x,\mathbf{n}}|^2 c_{\mathbf{n}}^2 \mathrm{e}^{-\mu\gamma/2} \prod_{|i|\le x}(1-\mathrm{e}^{-\mu}) \mathrm{e}^{-\mu (n_i+n_i^\prime)/2} .
\end{align}
At this point, we have two factors --  $\tilde{\mathcal{O}}_{x,\mathbf{n}}$ and $c_{\mathbf{n}}$ -- that must be bounded.  First, we use (\ref{eq:nae}) and (\ref{eq:cnbound}) to show that \begin{align}
        c_{\mathbf{n}}^2 \prod_{|i|\le x} \mathrm{e}^{-\mu (n_i+n_i^\prime)/4} < \mathrm{e}^{\mu\beta/4} \left(\frac{4\beta}{\mathrm{e}\mu}\right)^\beta := C_2. \label{eq:C2def}
\end{align}
Then, we can use (\ref{eq:68ab}) and (\ref{eq:C2def}) to show that 
\begin{align}\label{eq:OOxOOx}
    \left(\mathcal{O}^\prime \mathcal{O}_x| \mathcal{O}^\prime \mathcal{O}_x \right)_\mu &\le C_2\mathrm{e}^{-\mu\gamma/2} \sum_{\mathbf{n} \in \mathbf{n}_{\le x}} |\tilde{\mathcal{O}}_{x,\mathbf{n}}|^2 \prod_{|i|\le x}(1-\mathrm{e}^{-\mu}) \mathrm{e}^{-\mu (n_i+n_i^\prime)/4} \nonumber\\ &\le C_2\mathrm{e}^{-\mu\gamma/2} \left(\frac{1-\mathrm{e}^{-\mu}}{1-\mathrm{e}^{-\mu/2}}\right)^{2x+1} C \left(\frac{vt}{x}\right)^{x/(2l+1)} \le 2C C_2 \mathrm{e}^{-\mu\gamma/2} \left(\frac{2^{4l+2}vt}{x}\right)^{x/(2l+1)}.
\end{align}
Again, similar manipulations follow for other operator orderings such as $\mathcal{O}_x\mathcal{O}^\prime$, and lead to an identical functional form up to a different choice of O(1) prefactors $C$ and $C_2$.

At this point, we are ready to invoke (\ref{eq:Aansatz}).  The key observation is that \begin{align}
    \mathrm{tr}\left(\sqrt{\tilde{\rho}} A^\dagger \sqrt{\tilde{\rho}} B \right) &\le \sqrt{\mathrm{tr}\left(\sqrt{\tilde{\rho}} A^\dagger \sqrt{\tilde{\rho}} A \right) \mathrm{tr}\left(\sqrt{\tilde{\rho}} B^\dagger \sqrt{\tilde{\rho}} B \right)} \nonumber\\ &\le K_0\theta^{2x} \sqrt{ \mathrm{tr}\left(\sqrt{\rho_\mu } A^\dagger \sqrt{\rho_\mu } A \right) \mathrm{tr}\left(\sqrt{\rho_\mu } B^\dagger \sqrt{\rho_\mu } B \right) }. \label{eq:rhoCS}
\end{align}
If we then expand out 
\begin{align}
    \left( \mathcal{O}^\prime (\mathbb{P}_R \mathcal{O}) |  \mathcal{O}^\prime (\mathbb{P}_R \mathcal{O}) \right)_{\tilde{\rho}} &= \left( \mathcal{O}^\prime \mathcal{O}_{\le r_+} |  \mathcal{O}^\prime \mathcal{O}_{\le r_+} \right)_{\tilde{\rho}} + \sum_{x>r_+} \mlr{\left( \mathcal{O}^\prime \mathcal{O}_x |  \mathcal{O}^\prime \mathcal{O}_{\le x} \right)_{\tilde{\rho}} + \mathrm{H.c.} } - \left( \mathcal{O}^\prime \mathcal{O}_x |  \mathcal{O}^\prime \mathcal{O}_x \right)_{\tilde{\rho}} \notag \\ &\le\left( \mathcal{O}^\prime \mathcal{O}_{\le r_+} |  \mathcal{O}^\prime \mathcal{O}_{\le r_+} \right)_{\tilde{\rho}} + \sum_{x>r_+} 2\left|\left( \mathcal{O}^\prime \mathcal{O}_x |  \mathcal{O}^\prime \mathcal{O}_{\le x} \right)_{\tilde{\rho}} \right| + \left( \mathcal{O}^\prime \mathcal{O}_x |  \mathcal{O}^\prime \mathcal{O}_x \right)_{\tilde{\rho}},
\end{align}
for each term above, we can bound it using (\ref{eq:rhoCS}):
\begin{align}
    &\left( \mathcal{O}^\prime (\mathbb{P}_R \mathcal{O}) |  \mathcal{O}^\prime (\mathbb{P}_R \mathcal{O}) \right)_{\tilde{\rho}} \le K_0\theta^{2r_+} \left(\mathcal{O}^\prime \mathcal{O}_{\le r_+}| \mathcal{O}^\prime \mathcal{O}_{\le r_+} \right)_\mu + K_0\sum_{x>r_+}\theta^{2x} \left[ 2\sqrt{ \left(\mathcal{O}^\prime \mathcal{O}_{ x}| \mathcal{O}^\prime \mathcal{O}_{ x} \right)_\mu \left(\mathcal{O}^\prime \mathcal{O}_{\le x}| \mathcal{O}^\prime \mathcal{O}_{\le x} \right)_\mu } + \left(\mathcal{O}^\prime \mathcal{O}_{ x}| \mathcal{O}^\prime \mathcal{O}_{x} \right)_\mu \right] \nonumber\\ &\le 2K_0 C_2\mathrm{e}^{-\mu\gamma/2} \left\{ \theta^{2r_+} C^\prime  2^{2r_+} \left(\frac{vt}{r}\right)^{r/(2l+1)} + \sum_{x>r_+}(2\theta)^{2x} \left[2\sqrt{C C^\prime} \left(\frac{vt}{r}\right)^{r/2(2l+1)} \left(\frac{vt}{x}\right)^{x/2(2l+1)} + C\left(\frac{vt}{x}\right)^{x/(2l+1)} \right] \right\} \nonumber\\
    &\le C_1^\prime \left(\frac{(2\theta)^{4l+2}v t}{r}\right)^{r/(2l+1)}, \;\;\; \text {if} \;\;\; (2\theta)^{8l+4}(1+\epsilon)v t<r, \label{eq:rhoxsum}
\end{align}
where $\epsilon>0$ is any finite constant, and $0<C_1^\prime<\infty$ is a constant independent of $r$ or $t$, but dependent on $\epsilon$. To derive the last inequality above, we have to approximately re-sum the two $x$-dependent terms, which is where we will introduce $\epsilon$.  Observe that for $x>r_+$, we may write \begin{equation}\label{eq:sumx}
    \sum_{x\ge r_+}\left(\frac{vt}{x}\right)^{x/(4l+2)} < \sum_{x\ge r_+}\left(\frac{vt}{r_+}\right)^{r_+/(4l+2)} \times \left(\frac{1}{1+\epsilon}\right)^{(x-r_+)/(4l+2)} = \left(1 - \frac{1}{(1+\epsilon)^{1/(4l+2)}}\right)^{-1}\left(\frac{vt}{r_+}\right)^{r_+/(4l+2)}.
\end{equation}
The $\epsilon$-dependent prefactor ends up absorbed in the constant $C_1^\prime$.  Using this identity on both terms in the $x$-sum of (\ref{eq:rhoxsum}), and noting that $r<r_+$, we obtain the final inequality of (\ref{eq:rhoxsum}).   A slightly awkward feature of this equation is that our bound is super-exponentially small when the prefactor $C_1^\prime$ diverges: namely, the velocity which is suggested by the parenthetical expression does not match the speed of the light cone in which the expression is valid.  The presentation of the bound in (\ref{eq:7thm}) simply replaces $(2\theta)^{4l+2}v \rightarrow (2\theta)^{8l+4}(1+\epsilon)v $ so that the formula directly implies the region where the light cone is valid.

The theorem follows because the second term in (\ref{eq:triang_ine}) can be treated exactly the same way.\end{proof}

The following corollary demonstrates that the assumptions of the above theorem are sufficiently mild that they allow us to prove a finite velocity of information, as measured by \emph{all finite density matrix elements} of a commutator:

\begin{cor}\label{cor63}
Let $|\red{\psi_1}\rangle$ and $|\red{\psi_2}\rangle$ denote many-body states such that the maximal number of bosons on any site is $m$.   Then for any $m<\infty$, there exists a velocity $0<v_*<\infty$ and a constant $0<C<\infty$ such that for operators $\mathcal{O},\mathcal{O}^\prime$ obeying the assumptions of Theorem \ref{thm71}, \begin{equation}
    \left|\langle \red{\psi_1}| [\mathcal{O}(t),\mathcal{O}^\prime] |\red{\psi_2}\rangle \right| \le C \left(\frac{v_* t}{r}\right)^r. \label{eq:cor63}
\end{equation}
\end{cor}
\begin{proof}
The goal is to apply Theorem \ref{thm71} to the following three choices of $\tilde \rho$: \begin{subequations}
    \begin{align}
        \tilde\rho_1 &= |\red{\psi_1}\rangle\langle \red{\psi_1}|, \\
        \tilde\rho_2 &= |\red{\psi_2}\rangle\langle \red{\psi_2}|, \\
        \tilde\rho_3 &= |\psi\rangle\langle \psi|, \;\;\; \text{where} \;\;\; |\psi\rangle = 2^{-1/2}(|\red{\psi_1}\rangle + \mathrm{e}^{\mathrm{i}\phi}|\red{\psi_2}\rangle).
    \end{align}
\end{subequations}
Here $\phi$ is real.  To see why this would be helpful, observe that \begin{align}
    \mathrm{tr}\left( \sqrt{\tilde\rho}_3 A^\dagger\sqrt{\tilde\rho}_3 A \right) =& |\langle \psi| A|\psi\rangle|^2 = \frac{1}{4}\left|\langle \red{\psi_1}| A|\red{\psi_1}\rangle+\langle \red{\psi_2}| A|\red{\psi_2}\rangle+(\mathrm{e}^{\mathrm{i}\phi}\langle \red{\psi_1}| A|\red{\psi_2}\rangle+\text{c.c.})\right|^2 \le |\langle \red{\psi_1}| A|\red{\psi_1}\rangle|^2+|\langle \red{\psi_2}| A|\red{\psi_2}\rangle|^2\nonumber\\ & +2|\langle \red{\psi_1}| A|\red{\psi_2}\rangle|^2 \le \mathrm{tr}\left( \sqrt{\tilde\rho}_1 A^\dagger\sqrt{\tilde\rho}_1 A \right) + \mathrm{tr}\left( \sqrt{\tilde\rho}_2 A^\dagger\sqrt{\tilde\rho}_2 A \right) + 2|\langle \red{\psi_1}| A|\red{\psi_2}\rangle|^2. \label{eq:matrixelrhotilde}
\end{align}

Our goal is now to verify that Theorem \ref{thm71} holds for each of these three density matrices. \red{Expand $\psi_i,(i=1,2)$ in the boson number eigenbasis $|\mathbf{n}_{\le x}\rangle$ on sites $\le x$, \begin{align}
    |\psi_i\rangle = \sum_{\mathbf{n}_{\le x}} a_{i,\mathbf{n}_{\le x}}|\mathbf{n}_{\le x}\rangle \otimes |\psi_{i,\mathbf{n}_{\le x}}\rangle,
\end{align}
where $|\psi_{i,\mathbf{n}_{\le x}}\rangle$ are normalized states on sites $>x$, so that \begin{align}
    \sum_{\mathbf{n}_{\le x}} |a_{i,\mathbf{n}_{\le x}}|^2 =1.
\end{align}
Thus if $\mathbb{Q}_x A = A$, since $\langle \psi_{i,\mathbf{n}_{\le x}}|\psi_{j,\mathbf{n}^\prime_{\le x}}\rangle \le 1$ \begin{align}
    |\langle \red{\psi_i}| A|\red{\psi_j}\rangle|^2 &\le \left( \sum_{\mathbf{n}_{\le x},\mathbf{n}_{\le x}^\prime} |\bar{a}_{i,\mathbf{n}_{\le x}}a_{j,\mathbf{n}_{\le x}^\prime}| |\langle \mathbf{n}_{\le x}| A|\mathbf{n}^\prime_{\le x}\rangle|\right)^2 \le \sum_{\mathbf{n}_{\le x},\mathbf{n}_{\le x}^\prime} |a_{i,\mathbf{n}_{\le x}}|^2 |a_{j,\mathbf{n}_{\le x}^\prime}|^2 \sum_{\mathbf{n}_{\le x},\mathbf{n}_{\le x}^\prime} |\langle \mathbf{n}_{\le x}| A|\mathbf{n}^\prime_{\le x}\rangle|^2 \nonumber\\ &=\sum_{\mathbf{n}_{\le x},\mathbf{n}_{\le x}^\prime} |\langle \mathbf{n}_{\le x}| A|\mathbf{n}^\prime_{\le x}\rangle|^2 \le \mathrm{tr}\left( \sqrt{\rho_\mu} A^\dagger \sqrt{\rho_\mu} A \right)\times  \prod_{|j|\le x} \frac{\mathrm{e}^{\mu m}}{1-\mathrm{e}^{-\mu}}.
\end{align}
}
In the second line, simply observe that if the inner product is expanded out into all possible matrix elements of $A$, then when\begin{equation}
    \mu = \frac{1}{m},
\end{equation}the coefficient of $|\langle \mathbf{n}| A|\mathbf{n}^\prime\rangle|^2$ is greater than or equal to unity.  Of course, the second line includes all other possible matrix elements weighted by various factors.   We conclude that (\ref{eq:Aansatz}) holds for each of $\tilde\rho_{1,2,3}$ with $K_0 \le 4$ and  \begin{equation}
    \theta := \mathrm{e}\left(1+m\right)\ge \frac{\mathrm{e}}{1-\mathrm{e}^{-1/m}}.
\end{equation}
Therefore,
\begin{equation}
    \mathrm{tr}\left(\sqrt{\tilde\rho_{1,2,3}} [\mathcal{O}(t),\mathcal{O}^\prime]^\dagger\sqrt{\tilde\rho_{1,2,3}} [\mathcal{O}(t),\mathcal{O}^\prime] \right) < C^\prime \left(\frac{v_* t}{r}\right)^r \label{eq:rho123}
\end{equation}
for some constants $0<C^\prime,v_*<\infty$ as given in Theorem \ref{thm71}.  Combining (\ref{eq:matrixelrhotilde}) and (\ref{eq:rho123}), we obtain (\ref{eq:cor63}).
\end{proof}

Corollary \ref{cor63} provides a complete Lieb-Robinson-like bound for Bose-Hubbard-like models in one dimension.  Since we know, as discussed in the introduction, that a finite Lieb-Robinson velocity cannot exist in \emph{all states} as the physical velocity can diverge at high density, this is the strongest possible type of light cone.

Note that if $m\gg 1$ in Corollary \ref{cor63}, the velocity $v_*$ in (\ref{eq:cor63}) has parametrically different scaling at the (worst-case) density of $m$ than the bound for $\tilde{\rho} = \rho_{\mu}$ at $\mu=1/m$.   We believe this is not likely to be a physical effect, though of course a further investigation is worthwhile.

We do not know how as of yet to generalize Theorem \ref{thm71} or Corollary \ref{cor63} to a higher-dimensional lattice model.  The simple reason is that in $d$ dimensions, a ball of radius $r$ has $r^d$ sites inside, and so $\theta^{r^d}$ grows too quickly to merely ``rescale" the velocity of our light cone.  For any $d=2,3,\ldots$, the bound of \cite{kuwahara2021liebrobinson} can be better.  However, we note that the bound of \cite{kuwahara2021liebrobinson} requires that the density matrix $\tilde \rho$ commutes with the Hamiltonian $H$.  In general, we only expect a two-parameter family of such $\tilde \rho$ of broad physical importance:  $\tilde \rho \propto \exp[-\beta (H-\tilde \mu N)]$.  Our Theorem \ref{thm51} applies to this case whenever one considers the limit of $\beta=0$ and $\beta\tilde \mu := -\mu$ remains finite.

The methodology behind Corollary \ref{cor63} is not limited to this particular setting of interacting boson systems.  Indeed, it is easily generalized to prove that Frobenius and Lieb-Robinson light cones are (up to O(1) factors) equivalent in one dimensional models with local interactions: although this result was known previously \cite{Yin:2020pjd}, the current approach gives an alternative perspective as to why this must be the case.   In the presence of long-range interactions, however, it is known that the Frobenius and Lieb-Robinson light cones are distinct \cite{Tran:2020xpc}:  hence, it is possible to have a finite velocity for Frobenius commutator bounds, but diverging velocity for the usual operator norm of a commutator.  From the perspective of Theorem \ref{thm71}, this is allowable because the tail in the bounds is only algebraic:  $(t/r^\alpha)^\beta$ for some finite coefficients $\alpha$ and $\beta$.  Because $\beta$ does not scale with $r$, it is not generally possible to apply Theorem \ref{thm71} in these models without qualitatively changing the shape of the light cone, unless the number of sites on which the state is specified is $r$-independent.

\section{Classical complexity of simulations}\label{sec:complexity}
We can now prove that Bose-Hubbard-type models in one dimension are asymptotically no harder to simulate classically than usual spin chains.  This result provides mathematical justification to the routine simulation of low-density Bose gases by working in a truncated Hilbert space.  

More precisely, our results will bound the size of the finite dimensional Hilbert space needed to accurately calculate $\mathrm{tr}\lr{\tilde{\rho} \mathcal{O}(t)}$ using a classical computer, where for simplicity we assume $\mathbb{Q}_0|\mathcal{O})=|\mathcal{O})$. Although the finite density condition (\ref{eq:Aansatz}) is sufficient for our purpose, we use a potentially weaker version instead, assuming \begin{equation}\label{eq:rho<theta}
    \mathrm{tr}_{\le x} \left( \rho_{\mu,\le x}^{-1/2} (\mathrm{tr}_{> x}\tilde{\rho}) \rho_{\mu,\le x}^{-1/2} (\mathrm{tr}_{>x}\tilde{\rho}) \right) \le K_0\theta^{2x}, \quad \forall x,
\end{equation}
for some $\mu,\theta,K_0>0$, where $\rho_{\mu,\le x}:=\mathrm{tr}_{>x}\rho_\mu$.   In the above identity, the Hilbert space has been truncated to sites $\{-x,\cdots,x\}$, which is denoted with the appropriate subscripts. Similar to (\ref{eq:Aansatz}), (\ref{eq:rho<theta}) requires the boson density in $\tilde{\rho}$ to be at most of order $\theta$, as one can verify for boson number eigenstates $|\mathbf{n}\rangle$.  As a simple example, if in the initial state we know that there are exactly $N_{\le x}$ bosons on sites $\le x$, then we can choose \begin{equation}
   \mu = \frac{2x+1}{N_{\le x}}, \;\;\; K_0= \theta = \frac{\mathrm{e}}{1-\mathrm{e}^{-\mu}}.
\end{equation}
In Proposition \ref{prop83} at the end of this section, we show (\ref{eq:Aansatz}) implies (\ref{eq:rho<theta}) with a change of parameters $K_0, \theta$. 

Outlining the steps we need to obtain a bound on the computability..., we will first show in Theorem \ref{thm64} that $\mathcal{O}(t)$ can be approximated with exponential accuracy by \footnote{For notational convenience, we assume the Hamiltonian is time-independent.} \begin{equation}
    \mathcal{O}(t)_{\le r} := \mathrm{e}^{\mathrm{i}H_{\le r}t} \mathcal{O}\mathrm{e}^{-\mathrm{i}H_{\le r}t}.
\end{equation} Here $H_{\le r}$ is a Hamiltonian which acts only on sites $\le r+l$ with $r\gtrsim vt$ to be determined: \begin{align}
    H_{\le r} = \sum_{i=-r}^{r-1} H_{J,i} + \sum_{i=-r-l}^r U_i,
\end{align} where $H_{J,i}$ is hopping between $i,i+1$, and $U_i$ acts on $i,\cdots,i+l$.  (Recall that $l$ is the range of interactions, and $l=0$ for the Bose-Hubbard model.) Then, we will show in Proposition \ref{prop8} that (\ref{eq:rho<theta}) implies that the accurate calculation of $\mathrm{tr}(\tilde \rho \mathcal{O}(t))$ can be done in a finite-dimensional Hilbert space, with an error vanishing exponentially at large $t$.

\begin{thm}\label{thm64}
Let $\mathcal{O}$ be an operator on site 0 consisting of a finite product of creation and annihilation operators. If there are some $\mu, \theta,K_0>0$ such that $\tilde{\rho}$ satisfies (\ref{eq:rho<theta}), then the error by restricting to $H_{\le r}$ is bounded by \begin{equation}
    |\mathrm{tr}\lr{\tilde{\rho} \mathcal{O}(t)}-\mathrm{tr}\lr{\tilde{\rho} \mathcal{O}(t)_{\le r}}| \le C_3 rt\left(\frac{(2\theta)^{4l+2}v^\prime t}{r}\right)^{r/(4l+2)}.
\end{equation}
for $r>(2\theta)^{4l+2}v^\prime t$. Here $v^\prime = (1+\epsilon) v_{\mu/2}$ where $v_\mu$ is given in (\ref{eq:6v}) and $\epsilon$ is arbitrarily small but finite.  The constants $0<C_3,\epsilon <\infty$ are independent of $r$.
\end{thm}

\begin{proof}
Take $H$ to be time-independent for notational simplicity; however, the result holds for $t$-dependent $H$ as well with straightforward modifications. Decompose $H=J_{>r} + (H-J_{>r})$, where $J_{>r}=\sum_{i\ge r} H_{J,i}+\sum_{i<-r} H_{J,i}$ contains all the hopping terms in $H$ that are not included in $H_{\le r}$. Since all interaction terms commute, the evolution by $H-J_{>r}$ is \begin{align}\label{eq:H-J}
    \mathrm{e}^{-\mathrm{i}t(H-J_{>r})} = \mathrm{e}^{-\mathrm{i}t H_{\le r}} \mathrm{e}^{-\mathrm{i}t U_{>r} },
\end{align} 
where $U_{>r}$ contains interaction terms in $H$ that are not included in $H_{\le r}$. As a result, $\mathcal{O}$ evolved by $H-J_{>r}$ is the same as by $H_{\le r}$, so the error is expressed using Duhamel identity \begin{equation}\label{eq:Deltar}
    \Delta_r:=\mathrm{tr}\lr{\tilde{\rho} \mathcal{O}(t)}-\mathrm{tr}\lr{\tilde{\rho} \mathcal{O}(t)_{\le r}} = \mathrm{i}\int^t_0 \mathrm{tr}\lr{\tilde{\rho}\mathrm{e}^{\mathrm{i}s(H- J_{>r})}[J_{>r}, \mathcal{O}(t-s)]\mathrm{e}^{-\mathrm{i}s(H- J_{>r})}} \mathrm{d}s.
\end{equation}
Similar to (\ref{eq:Ox}), we can replace $\mathcal{O}(t-s)$ by \begin{equation}
    \mathbb{P}_{\ge r}\mathcal{O}(t-s) = \sum_{x=r}^L \mathcal{O}_x.
\end{equation} Commutator $[J_{>r}, \mathcal{O}_x]$ has no support beyond $x+1$, so that all interaction terms that act nontrivially outside $x+l+1$ do not contribute to evolution in (\ref{eq:Deltar}). Denote $(H-J_{>r})_x$ by dropping such terms in $H-J_{>r}$. Now view the evolution by $H-J_{>r}$ in (\ref{eq:Deltar}) as acting on $\tilde{\rho}$ instead.  Defining $\tilde{\rho}_x(s)=\mathrm{e}^{-\mathrm{i}s(H-J_{>r})_x}  \tilde{\rho}\mathrm{e}^{\mathrm{i}s(H-J_{>r})_x}$, we find that 
\begin{equation}
    \Delta_r =\mathrm{i}\int^t_0 \sum_{x\ge r}\mathrm{tr}\lr{\tilde{\rho}_x(s)\left[J_{>r}, \mathcal{O}_x\right]} \mathrm{d}s = \mathrm{i}\int^t_0 \sum_{x\ge r} \lr{Y_{x}|\left[ J_{>r}, \mathcal{O}_x\right]}_{\mu,\le x+l+1} \mathrm{d}s,
\end{equation}
where \begin{equation}
    Y_x:=\rho_{\mu,\le x+l+1}^{-1/2} (\mathrm{tr}_{> x+l+1}\tilde{\rho}_x(s)) \rho_{\mu,\le x+l+1}^{-1/2},
\end{equation}
and the inner product $(\cdot|\cdot)_{\mu,\le x+l+1}$ is taken assuming that the Hilbert space has support only on sites within $x+l+1$. 
Since $\tilde{\rho}_x$ is evolved within $x+l+1$ only, we can first partial trace out the sites $>x+l+1$, and then evolve with time: $\mathrm{tr}_{> x+l+1}\tilde{\rho}_x(s) =(\mathrm{tr}_{> x+l+1}\tilde{\rho})(s)$. Thus property (\ref{eq:rho<theta}) persists under such evolution, because $\rho_\mu$ is stationary: \begin{align}
    (Y_x|Y_x)_{\mu,\le x+l+1} &=\mathrm{tr}_{\le x+l+1} \left( \rho_{\mu,\le x+l+1}^{-1/2} (\mathrm{tr}_{> x+l+1}\tilde{\rho}_x(s)) \rho_{\mu,\le x+l+1}^{-1/2} (\mathrm{tr}_{> x+l+1}\tilde{\rho}_x(s)) \right) \nonumber\\ &= \mathrm{tr}_{\le x+l+1} \left( \rho_{\mu,\le x+l+1}^{-1/2} (\mathrm{tr}_{> x+l+1}\tilde{\rho})(s)\rho_{\mu,\le x+l+1}^{-1/2} (\mathrm{tr}_{> x+l+1}\tilde{\rho})(s) \right) \nonumber\\ &= \mathrm{tr}_{\le x+l+1} \left( \rho_{\mu,\le x+l+1}^{-1/2} (\mathrm{tr}_{> x+l+1}\tilde{\rho}) \rho_{\mu,\le x+l+1}^{-1/2} (\mathrm{tr}_{>x+l+1}\tilde{\rho}) \right) \le K_0\theta^{2(x+l+1)}, \quad \forall x\ge r.
\end{align}
Furthermore, expand $J_{>r}$ by $H_{J,i}$ and use (\ref{eq:OOxOOx}) with $\mathcal{O}^\prime=H_{J,i}$ and $\beta=2,\gamma=0$, \begin{align}
    |\Delta_r| &\le 2t \sum_{x\ge r}\sum_{i=r}^x \sqrt{(Y_x|Y_x)_{\mu,\le x+l+1}}\lr{\sqrt{(H_{J,i}\mathcal{O}_x|H_{J,i}\mathcal{O}_x)_{\mu,\le x+l+1}} + \sqrt{(\mathcal{O}_xH_{J,i}|\mathcal{O}_xH_{J,i})_{\mu,\le x+l+1}}} \nonumber\\ &\le 4t \sum_{x\ge r} x \sqrt{2K_0\theta^{2(x+l+1)} C C_2 \left(\frac{2^{4l+2}vt}{x}\right)^{x/(2l+1)}} \le C_3 rt \left(\frac{(2\theta)^{4l+2}vt}{r}\right)^{r/(4l+2)},
\end{align}
where the factor of $2$ in the first line acounts for both directions, and $C_3$ is a constant independent of $r$ or $t$. We have assumed $r>(2\theta)^{4l+2}(1+\epsilon)v$ with any finite constant $\epsilon >0$ to get the last equation, using manipulations similar to (\ref{eq:sumx}).
\end{proof}

\begin{prop}
\label{prop8}
Under the conditions of Theorem \ref{thm64}, set \begin{align}
    r= \mathrm{e}(2\theta)^{4l+2} v' t,
\end{align} and let $\tilde \rho_{\le N_0}$ denote the restriction of $\mathrm{tr}_{> r+l}\tilde \rho$ to the Hilbert space of states with $\le N_0$ bosons on sites $|i|\le r+l$.  $\tilde \rho_{\le N_0}$ does not need to be normalized. Choose \begin{align}\label{eq:N0def}
    N_0 = (2r+2l+1)\max\left(4, \frac{1}{\mathrm{e}^{\mu/3}-1 },\frac{2}{\mu}\left(1+8\ln2+\ln\frac{\theta(1-\mathrm{e}^{-\mu})}{\mu^4}\right)\right),
\end{align} then the error of calculating the dynamics of $\mathcal{O}$ by restricting to this Hilbert space is bounded by \begin{equation}\label{eq:prop8}
    |\mathrm{tr}\lr{\tilde{\rho} \mathcal{O}(t)}-\mathrm{tr}\lr{\tilde{\rho}_{\le N_0} \mathcal{O}(t)_{\le r}}| \le C_4 r^2 \mathrm{e}^{-r/(4l+2) },
\end{equation}
where $0<C_4<\infty$ is independent of $r$.
\end{prop}

\begin{proof}
Using the triangle inequality, \begin{align}
    |\mathrm{tr}\lr{\tilde{\rho} \mathcal{O}(t)}-\mathrm{tr}\lr{\tilde{\rho}_{\le N_0} \mathcal{O}(t)_{\le r}}| \le |\mathrm{tr}\lr{\tilde{\rho} \mathcal{O}(t)}-\mathrm{tr}\lr{\tilde{\rho}\mathcal{O}(t)_{\le r}}| + |\mathrm{tr}\lr{\tilde{\rho} \mathcal{O}(t)_{\le r}}-\mathrm{tr}\lr{\tilde{\rho}_{\le N_0} \mathcal{O}(t)_{\le r}}|,
\end{align}
where the first term is bounded by the form of the right hand side of (\ref{eq:prop8}) according to Theorem \ref{thm64}. Thus we only need to bound the second term by the same form. Since it only involves dynamics within sites $|i|\le r+l$, for the rest of this proof we can denote with $\tilde{\rho}$ the initial state restricted to this segment of length $r':=2r+2l+1$. Its support at large boson numbers is bounded by (\ref{eq:rho<theta}): \begin{equation}
    \sum_{N,N'} \mathrm{e}^{\mu(N+N')/2} \sum_{\mathbf{n}\in \mathbf{n}_N, \mathbf{n}'\in \mathbf{n}_{N'}} |\tilde{\rho}_{\mathbf{n}\mathbf{n}'}|^2 \le \left(\theta(1-\mathrm{e}^{-\mu})\right)^{r^\prime},
\end{equation} 
where $\mathbf{n}_N$ is the set of all $\mathbf{n}$ with total boson number $N$. By counting the number of ways to arrange $N$ bosons on $r^\prime$ sites (a textbook statistical mechanics problem), we find \begin{equation}\label{eq:|n_N|}
    |\mathbf{n}_N|\le\frac{(N+r^\prime)!}{r^\prime!N!}\le \mathrm{exp}\left[r^\prime\ln\left(1+\frac{N}{r^\prime}\right) + N\ln\left(1+\frac{r^\prime}{N}\right) \right].
\end{equation}

By the assumptions in the proposition, the matrix element $\langle \mathbf{n}|\mathcal{O}(t)|\mathbf{n}'\rangle$ is nonzero only for $N-N'=\gamma\ge 0$, and is bounded by a power $K_\mathcal{O}N^\beta$, where $\beta,\gamma$ are given in (\ref{eq:beta_gamma}). Using this to argue why the $N_0$ chosen in (\ref{eq:N0def}) is useful, observe that the error of this truncation is: \begin{align} \label{eq:75}
    &|\mathrm{tr}\left[(\tilde{\rho}-\tilde{\rho}_{\le N_0}) \mathcal{O}(t)\right]| \le \sum_{N> N_0} \sum_{\mathbf{n}\in \mathbf{n}_N, \mathbf{n}'\in \mathbf{n}_{N-\gamma}} |\tilde{\rho}_{\mathbf{n}\mathbf{n}'} | K_\mathcal{O} N^\beta \le K_\mathcal{O}\sum_{N> N_0} N^\beta \sqrt{|\mathbf{n}_N||\mathbf{n}_{N-\gamma}| \sum_{\mathbf{n}\in \mathbf{n}_N, \mathbf{n}'\in \mathbf{n}_{N-\gamma}} |\tilde{\rho}_{\mathbf{n}\mathbf{n}'}|^2} \nonumber\\ &\le K_\mathcal{O}\sum_{N> N_0} N^\beta \frac{(N+r^\prime)!}{r^\prime!N!} \mathrm{e}^{-\mu(N-\gamma/2)/2} \left(\theta(1-\mathrm{e}^{-\mu})\right)^{r^\prime/2} := K_\mathcal{O}\mathrm{e}^{\mu\gamma/4} \sum_{N>N_0}q_N.
\end{align}
In the last step we defined the sequence $(q_N)_{N\ge N_0}$. Since for any $0<\epsilon^\prime<(1-\mathrm{e}^{-\mu/2})/2$, \begin{align}
    \frac{q_{N}}{q_{N-1}} = \left(\frac{N}{N-1}\right)^\beta \frac{N+r^\prime}{N} \mathrm{e}^{-\mu/2}\le 1-\epsilon^\prime, \quad \text{if} \quad \left\{\begin{array}{ccc}
        \left(\frac{N}{N-1}\right)^\beta\le \frac{1}{1-\epsilon^\prime} & \Leftrightarrow & N\ge \frac{1}{1-(1-\epsilon^\prime)^{1/\beta}}, \\
        \frac{N+r^\prime}{N} \mathrm{e}^{-\mu/2}\le 1-2\epsilon^\prime & \Leftrightarrow & N\ge \frac{r^\prime}{\mathrm{e}^{\mu/2}(1-2\epsilon^\prime) -1},
    \end{array} \right.
\end{align} the sequence is bounded by an exponential: $q_N\le q_{N_0} (1-\epsilon^\prime)^{N-N_0}$, if \begin{equation}
    \theta^\prime:=\frac{N_0}{r^\prime}\ge \frac{1}{\mathrm{e}^{\mu/2}(1-2\epsilon^\prime) -1}, \label{eq:82ineq1}
\end{equation} and \begin{equation}
    r^\prime > \frac{\mathrm{e}^{\mu/2}(1-2\epsilon^\prime) -1}{1-(1-\epsilon^\prime)^{1/\beta}}. \label{eq:82ineq2}
\end{equation}
When inequalities (\ref{eq:82ineq1}) and (\ref{eq:82ineq2}) hold, the sum over $q_N$ is bounded by $q_{N_0}/\epsilon^\prime$. Then bounding the binomial coefficient by (\ref{eq:|n_N|}),  \begin{equation}\label{eq:e^r'}
   |\mathrm{tr}\left[(\tilde{\rho}-\tilde{\rho}_{\le N_0}) \mathcal{O}(t)\right]|  \le K_\mathcal{O}\mathrm{e}^{\mu\gamma/4}(r^\prime\theta^\prime)^\beta \frac{1}{\epsilon^\prime} \mathrm{exp}\left\{\frac{r^\prime}{2}\left[\ln(\theta(1-\mathrm{e}^{-\mu})) + 2(1+\theta^\prime)\ln(1+\theta^\prime)-2\theta^\prime\ln\theta^\prime-\mu\theta^\prime\right]\right\}.
\end{equation}
The error can be made exponentially small in $r^\prime$ by choosing a  $\theta^\prime$ determined by $\theta,\mu$, but not $r$. To be concrete, we restrict to the case $\theta'\ge 4$ so that \begin{align}
    (1+\theta^\prime)\ln(1+\theta^\prime)-\theta^\prime\ln\theta^\prime < 2\ln\theta^\prime.
\end{align}
We then wish to satisfy
\begin{align}\label{eq:tilde_mu}
    \ln\xi+ 4\ln\theta^\prime -\mu\theta^\prime \le 0 \quad \Leftrightarrow \quad \ln (\xi^{1/4}\theta^\prime) \le \tilde{\mu}(\xi^{1/4}\theta^\prime), \quad \tilde{\mu}:=\frac{\mu}{4}\xi^{-1/4},
\end{align}
where we have set $\xi=\mathrm{e}\theta(1-\mathrm{e}^{-\mu})$ so that the exponent in (\ref{eq:e^r'}) is smaller than $-r^\prime/2$. If $\tilde{\mu}\ge 1/\mathrm{e}$, (\ref{eq:tilde_mu}) holds for any $\theta^\prime>0$; otherwise one can verify $\xi^{1/4}\theta'\ge \frac{2}{\tilde{\mu}}\ln\frac{1}{\tilde{\mu}}>2\mathrm{e}$ suffices. Thus considering (\ref{eq:82ineq1}) in addition, we choose \begin{align}
    \theta'=\max\left(4,\xi^{-1/4}\frac{2}{\tilde{\mu}}\ln\frac{1}{\tilde{\mu}}, \frac{1}{\mathrm{e}^{\mu/3}-1 }\right) = \max\left(4, \frac{1}{\mathrm{e}^{\mu/3}-1 },\frac{2}{\mu}\left(1+8\ln2+\ln\frac{\theta(1-\mathrm{e}^{-\mu})}{\mu^4}\right)\right),
\end{align}
where we have set $2\epsilon^\prime = 1-\mathrm{e}^{-\mu/6}$. Such $\theta^\prime$ makes the error exponentially small: \begin{align}
    |\mathrm{tr}\left[(\tilde{\rho}-\tilde{\rho}_{\le N_0}) \mathcal{O}(t)\right]|  \le K_\mathcal{O}\mathrm{e}^{\mu\gamma/4}(r^\prime\theta^\prime)^\beta \frac{2}{1-\mathrm{e}^{-\mu/6}} \mathrm{e}^{-r^\prime/2},
\end{align}
which can be massaged to the form of the right hand side of (\ref{eq:prop8}). This completes the proof.
\end{proof}

This Proposition rigorously proves that it is not asymptotically  harder to simulate the 1d Bose-Hubbard model at finite-density than it is to simulate any 1d model of interacting spins or fermions.  To simulate expectation value of a local observable for time $t$, with asymptotically vanishing error one could truncate the Hilbert space according to Proposition \ref{prop8}.
Since this truncated Hilbert space has dimension $D$ obeying \begin{equation}
    \log D \lesssim \theta^\prime r^\prime \propto t,
\end{equation}, we find that the dynamics can be simulated with $\mathrm{exp}(\mathrm{O}(t))$ classical resources. Consider separating the whole time region time steps of size $t_0$. At each step, a naive discretization for evolving the density matrix $\tilde{\rho}$ induces an error \begin{align}
    \lVert \delta\tilde{\rho}\rVert_1 =\mathrm{O}(\lVert (Ht_0)^2 \tilde{\rho}  \rVert_1) =\mathrm{O}(\norm{H}^2t_0^2).
\end{align}
The total error for all steps is multiplied by an extra factor $t/t_0$.  If we desire the error in $\langle \mathcal{O}(t) \rangle$ to be at most $\epsilon$, then we need \begin{equation}
   \norm{\mathcal{O}} \lVert \delta\tilde{\rho}\rVert_1 \times \frac{t}{t_0} = \norm{\mathcal{O}}\frac{t}{t_0}\mathrm{O}(\norm{H}^2t_0^2) \le \epsilon,
\end{equation}
which implies that \begin{align}
\frac{t}{t_0} = \mathrm{O}(\mathrm{poly}(t)\frac{1}{\epsilon}),
\end{align}
where we have used $\norm{\mathcal{O}},\norm{H}= \mathrm{O}(\mathrm{poly}(t))$ thanks to the truncation. Since each step needs $\mathrm{poly}(D)$ resources, the total computational resources required are $\mathrm{exp}(\mathrm{O}(t))/\epsilon$.

Lastly, as advertised, let us show that (\ref{eq:Aansatz}) implies (\ref{eq:rho<theta}) Note that we do need version (\ref{eq:Aansatz}) to bound the ``2-norm'' of a growing operator in the previous section, while here we needed its ``$\infty$-norm'' in the ``finite density subspace".  More precisely,
\begin{prop}\label{prop83}
If the state $\tilde{\rho}$ satisfies (\ref{eq:Aansatz}), then \begin{equation}\label{eq:prop81}
    \mathrm{tr}_{\le x} \left( \rho_{\mu,\le x}^{-1/2} (\mathrm{tr}_{> x}\tilde{\rho}) \rho_{\mu,\le x}^{-1/2} (\mathrm{tr}_{>x}\tilde{\rho}) \right) \le \left(\frac{K_0\sqrt{1-\mathrm{e}^{-\mu}}}{1-\mathrm{e}^{-\mu/2}}\right)^2\left(\frac{\theta\sqrt{1-\mathrm{e}^{-\mu}}}{1-\mathrm{e}^{-\mu/2}}\right)^{4x}, \quad \forall x.
\end{equation}
\end{prop}

\begin{proof}
For a given $x$, denote the matrix elements of $\rho_{\mu,\le x}^{-1/2}$ by $\eta_{\mathbf{n}} := \langle\mathbf{n}| \rho_{\mu,\le x}^{-1/2}|\mathbf{n}\rangle$, where the index $\mathbf{n}$ runs over the boson number eigenstate basis on sites $\le x$. We first bound the left hand side of (\ref{eq:rho<theta}): \begin{align}\label{eq:rho<rho2}
    \mathrm{tr} \left( \rho_{\mu,\le x}^{-1/2} (\mathrm{tr}_{> x}\tilde{\rho}) \rho_{\mu,\le x}^{-1/2} (\mathrm{tr}_{>x}\tilde{\rho}) \right) = \sum_{\mathbf{n}\mathbf{n}^\prime} \eta_\mathbf{n}\eta_{\mathbf{n}^\prime} |\langle\mathbf{n}| \tilde{\rho}_{\le x}|\mathbf{n}^\prime\rangle|^2 \le \sum_{\mathbf{n}\mathbf{n}^\prime} \eta_\mathbf{n}\eta_{\mathbf{n}^\prime} \langle\mathbf{n}| \tilde{\rho}_{\le x}|\mathbf{n}\rangle \langle\mathbf{n}^\prime| \tilde{\rho}_{\le x}|\mathbf{n}^\prime\rangle = \left(\mathrm{tr} \rho_{\mu,\le x}^{-1/2} \tilde{\rho}_{\le x} \right)^2.
\end{align} Here and for the rest of this proof, we drop the subscript $\le x$ on the trace for simplicity.
Suppose there is a set of operators $\{A_p:p=0,1,\cdots\}$ which are supported on sites $\le x$, such that \begin{align}\label{eq:Ap}
    \sum_p A_p^\dagger \mathcal{O} A_p = \frac{1}{2} \{\mathcal{O}, \rho_{\mu,\le x}^{-1/2}\otimes I_{>x}\}
\end{align}
for any operator $\mathcal{O}$. Then choosing $\mathcal{O} = \sqrt{\tilde{\rho}}$, the root of the right hand side of (\ref{eq:rho<rho2}) is \begin{align}\label{eq:sump}
    \mathrm{tr} \rho_{\mu,\le x}^{-1/2} \tilde{\rho}_{\le x} = \frac{1}{2} \mathrm{tr}\left( \sqrt{\tilde{\rho}}\{\sqrt{\tilde{\rho}}, \rho_{\mu,\le x}^{-1/2}\otimes I_{>x}\}\right) = \sum_p \mathrm{tr}\left( \sqrt{\tilde{\rho}}A_p^\dagger \sqrt{\tilde{\rho}}A_p\right) \le K_0\theta^{2x} \sum_p \mathrm{tr}\left( \sqrt{\rho_\mu}A_p^\dagger \sqrt{\rho_\mu}A_p\right),
\end{align}
where (\ref{eq:Aansatz}) is used in the last step. Now we construct $\{A_p\}$ to evaluate the right hand side of (\ref{eq:sump}). To satisfy (\ref{eq:Ap}), we decompose $\mathcal{O}$ as $\mathcal{O} = \sum_{\mathbf{n}\mathbf{n}^\prime} \mathcal{O}_{\mathbf{n}\mathbf{n}^\prime}|\mathbf{n}\rangle \langle\mathbf{n}^\prime| \otimes \mathcal{O}_{\mathbf{n}\mathbf{n}^\prime}^\prime$, where each $\mathcal{O}_{\mathbf{n}\mathbf{n}^\prime}^\prime$ acts outside $x$. Since (\ref{eq:Ap}) is linear in $\mathcal{O}$ and the $> x$ parts of the operators on both sides agree trivially, it suffices to restrict to the $\le x$ sites and only consider the operator basis $\mathcal{O}=|\mathbf{n}\rangle \langle\mathbf{n}^\prime|$. Using ansatz $A_p^\dagger=A_p = \sum_\mathbf{n} A_{p\mathbf{n}} |\mathbf{n}\rangle \langle\mathbf{n}|$, (\ref{eq:Ap}) yields \begin{align}\label{eq:Apn}
    \sum_p A_{p\mathbf{n}}A_{p\mathbf{n}^\prime} = (\eta_\mathbf{n}+\eta_{\mathbf{n}^\prime})/2.
\end{align}
The left hand side can be viewed as the inner product between two vectors $A_{\cdot,\mathbf{n}}$ and $A_{\cdot,\mathbf{n}^\prime}$, so that all the vectors $A_{\cdot,\mathbf{n}}$ can be constructed inductively. For example, start from $A_{p\mathbf{0}}=\eta_\mathbf{0}\delta_{p0}$. To find a second vector $A_{\cdot,\mathbf{n}}$ with arbitrary $\mathbf{n}\neq \mathbf{0}$, we set $A_{p\mathbf{n}}=0$ for all $p>1$. The only nonzero elements $A_{0\mathbf{n}},A_{1\mathbf{n}}$ are then determined by the inner product with $A_{\cdot,\mathbf{0}}$ and with itself using (\ref{eq:Apn}). However, the specific form of $A_{p\mathbf{n}}$ is not important for the proof, and we only need its existence. (\ref{eq:Apn}) then implies \begin{align}
    \sum_p \mathrm{tr}\left( \sqrt{\rho_\mu}A_p^\dagger \sqrt{\rho_\mu}A_p\right) = \sum_p \sum_{\mathbf{n}} A_{p\mathbf{n}}A_{p\mathbf{n}} \eta_\mathbf{n}^{-2}= \sum_{\mathbf{n}} \eta_\mathbf{n}^{-1} = \mathrm{tr}_{\le x} \sqrt{\rho_{\mu,\le x}} = \left(\frac{\sqrt{1-\mathrm{e}^{-\mu}}}{1-\mathrm{e}^{-\mu/2}}\right)^{2x+1}.
\end{align}
Finally, (\ref{eq:prop81}) follows by combining the above equation with (\ref{eq:rho<rho2}) and (\ref{eq:sump}).
\end{proof}

\section{Clustering of correlations in the ground state}\label{sec:gs}

As another application of Theorem \ref{thm71}, we prove exponential clustering for gapped ground states in one dimension in any model of interacting bosons described by Hamiltonians with density-dependent interactions. In particular, we assume there is a nondegenerate ground state $|E_0\rangle$. Let $\mathcal{O},\mathcal{O}^\prime$ be two operators which are supported on two sets of sites, whose supports are separated by distance $r$. They can be unbounded operators such as $b$s or $b^\dagger$s; we only require their actions on the ground state do not lead to states with unbounded norm (this property will be satisfied by products of $b$ or $b^\dagger$ if $|E_0\rangle$ has bounded boson number on each site): $\norm{\mathcal{O}|E_0\rangle}_2\norm{\mathcal{O}^\prime|E_0\rangle}_2<\infty$.  Define their ground state correlation as \begin{align}
    \mathrm{Cor}(\mathcal{O},\mathcal{O}^\prime) := \langle E_0|\mathcal{O}\mathcal{O}^\prime|E_0\rangle - \langle E_0|\mathcal{O}|E_0\rangle\langle E_0|\mathcal{O}^\prime|E_0\rangle.
\end{align}
If the ground state density matrix satisfies condition (\ref{eq:Aansatz}), the following theorem proves that this correlation decays exponentially with the operators' separation $r$, whenever there is a finite gap to the first excited state. 

\begin{thm}\label{thm9}
Let $H$ be a time-independent Hamiltonian.  Assume there is a nondegenerate ground state $\tilde{\rho}=|E_0\rangle\langle E_0|$ satisfying (\ref{eq:Aansatz}).   Let $\Delta E$ be the spectral gap of $H$.  Then whenever $\norm{\mathcal{O}|E_0\rangle}\norm{\mathcal{O}^\prime|E_0\rangle} < \infty$ and the support of $\mathcal{O}$ and $\mathcal{O}^\prime$ is separated by $r$, \begin{align}\label{eq:9thm}
    |\mathrm{Cor}(\mathcal{O},\mathcal{O}^\prime)| \le C_5 \mathrm{exp}\left(-\frac{\Delta E}{2v}r\right),
\end{align}
where $0<C_5<\infty$ is independent of $r$, and $v=(2\theta)^{8l+4}v^\prime$ is given in Theorem \ref{thm71}.
\end{thm}

\begin{proof}
    The proof follows earlier work such as \cite{Hastings:2005pr,kuwahara2021liebrobinson}.  Without loss of generality, set $E_0=0$.  Consider the identity (e.g. (S.29) in \cite{kuwahara2021liebrobinson}) \begin{align}\label{eq:cor=}
        \mathrm{Cor}(\mathcal{O},\mathcal{O}^\prime) = \int\limits_{-T}^{T} \mathrm{d}t K(t)\left\langle E_{0}\left|\left[\mathcal{O}(t), \mathcal{O}^\prime\right]\right| E_{0}\right\rangle - \left(\langle E_{0}|\mathcal{O} Q_{T} \mathcal{O}^\prime| E_{0}\rangle+\text {c.c.}\right),
    \end{align}
    where the parameter $T$ is to be determined, and \begin{align}
        K(t):=\frac{\mathrm{i}}{2 \pi} \lim _{\epsilon \rightarrow0^+} \frac{\mathrm{e}^{-\frac{\Delta E t^{2}}{2T}}}{t+\mathrm{i} \epsilon}, \quad Q_T:= \sum_{s\ge 1} \tilde{K}(E_s)|E_s\rangle\langle E_s|. \end{align}
    Here $c$ is an O(1) constant, $\tilde{K}(E)$ is the Fourier transform of $K(t)$, and $\{|E_s\rangle: s\ge 0\}$ denote all the eigenstates. Note that
    \begin{align}
        \lVert Q_T\rVert \le \max_{s\ge 1} |\tilde{K}(E_s)| \le \frac{c}{2} \mathrm{e}^{-T\Delta E/2}.
    \end{align}
    It follows that \begin{equation}
        \langle E_{0}|\mathcal{O} Q_{T} \mathcal{O}^\prime| E_{0}\rangle+\text {c.c.} \le \norm{\mathcal{O}|E_0\rangle}\norm{\mathcal{O}^\prime|E_0\rangle} \mathrm{e}^{-T\Delta E/2}.
    \end{equation}
    We now use Theorem \ref{thm71} to bound the first term. \begin{align}\label{eq:OO<}
        |\mathrm{tr}(\tilde{\rho} [\mathcal{O}(t), \mathcal{O}^\prime])| \le \left[ \mathrm{tr}\left(\sqrt{\tilde{\rho}}[\mathcal{O}(t),\mathcal{O}^\prime]^\dagger \sqrt{\tilde{\rho}} [\mathcal{O}(t),\mathcal{O}^\prime] \right)\right]^{1/2} \le \sqrt{C_1} \left(\frac{v t}{r}\right)^{r/(4l+2)},
    \end{align}
    where $v=(2\theta)^{8l+4}v^\prime$ as in (\ref{eq:7thm}). Then \begin{align}
        \left|\int_{-T}^{T} K(t)\left\langle E_{0}\left|\left[\mathcal{O}(t), \mathcal{O}^\prime\right]\right| E_{0}\right\rangle d t\right| \le \sqrt{C_1}\left(\frac{v T}{r}\right)^{r/(4l+2)-1} \frac{v}{r} \int_{-T}^{T} |tK(t)|d t \le \sqrt{\frac{C_1}{\pi \Delta E}}\left(\frac{v T}{r}\right)^{r/(4l+2)}.
    \end{align}
    Combining both terms in (\ref{eq:cor=}) and choosing \begin{align}
        T=\frac{r}{v}\mathrm{exp}\left(-(4l+2)\frac{\Delta E}{2v}\right)<\frac{r}{v},
    \end{align}
    we obtain \begin{align}
        |\mathrm{Cor}(\mathcal{O},\mathcal{O}^\prime)| \le \left(\sqrt{\frac{C_1}{\pi \Delta E}} + c \norm{\mathcal{O}|E_0\rangle}\norm{\mathcal{O}^\prime|E_0\rangle} \right)\mathrm{exp}\left(-\frac{\Delta E}{2v}r\right),
    \end{align}
    which reduces to (\ref{eq:9thm}).
\end{proof}

(\ref{eq:9thm}) improves a previous recent result in  \cite{kuwahara2021liebrobinson}, where both the form of operators $\mathcal{O},\mathcal{O}^\prime$ are more restricted, and the bound on the correlation decays subexponentially as $\mathrm{exp}(-c\sqrt{r/\ln r})$.  Our improvement arises due to the tight tails in our linear light cone.

Generalizing this result to models with degenerate ground states appears straightforward (see e.g. \cite{Hastings:2005pr}).  A standard application of clustering theorems for finite-dimensional quantum systems has been the proof of an entanglement area law.   In the bosonic case, this appears to be more subtle because all existing bounds have explicit constants which depend on the Hilbert space dimension \cite{brandao2013area,cho18area}.  Thus, it is not entirely straightforward to use Theorem \ref{thm9} to prove an entanglement area law for bosonic models.  Nevertheless we anticipate that further generalizing the results of Section \ref{sec:complexity}, it may be possible to project $|E_0\rangle\langle E_0|$ into a subspace with bounded boson number on each site, at which case Theorem \ref{thm9} would also lead to an entanglement area law.



\section{Outlook}
 Inspired by earlier work \cite{Nachtergaele_2008,kuwahara2021liebrobinson}, we have proven that correlators and out-of-time-ordered correlators, measured in the infinite temperature grand canonical ensemble defined in (\ref{eq:rho}), vanish outside of a ``linear light cone" in a broad family of interacting boson models:   $\langle [A_0(t),B_r]\rangle \rightarrow 0$ if $vt<r$, with asymptotics encapsulated in (\ref{eq:main}).   As we highlighted in the introduction, our bound on $v$  is qualitatively optimal for the Bose-Hubbard model, for commutators involving single boson creation/annihilation operators in any dimension, and for all commutators in one dimension.  In one dimension, we generalized this result to prove that all matrix elements of a commutator $\langle \psi_1|[A_0(t),B_r] |\psi_2\rangle$ are vanishingly small outside of a light cone $v^\prime t < r$, with a slightly larger velocity $v^\prime$, whenever the states $|\psi_{1,2}\rangle$ have a bounded number of bosons on each site.  This latter result could then be used to demonstrate the computational complexity of simulating Bose gases, along with the exponential decay of correlations in gapped ground states, in a broad range of experimentally relevant states of one-dimensional Bose gases.

We hope that our result will be generalized in important directions.  Firstly, (though of less general interest), we anticipate likely order of magnitude improvements in the O(1) coefficients in our bound (\ref{eq:thm2}).   Secondly and more importantly, we were not able to prove that all local correlators are bounded by the \emph{same velocity}, outside of one dimensional models.  We believe this to be a physically reasonable property, yet the quantum walk formalism we developed is not sufficiently developed to prove this property, which may rely on some more sophisticated clustering approximations.  We hope that this issue can be resolved in the near future.   Thirdly, we have only proven in one spatial dimension that there does not exist \emph{any} finite density state where quantum information cannot spread with arbitrarily large velocity.  In higher dimensions, our bound only shows that such states are vanishingly rare in the grand canonical ensemble.  The technical reason why we were unable to prove that no such state with ``superluminal" propagation can exist is essentially that the density matrix $\rho$ defined in (\ref{eq:rho}) is unique in that it commutes with all number-conserving $H(t)$ \emph{and} is a tensor product: namely, the density matrix has a strict form of locality.  These properties of $\rho$ are crucial to the anti-Hermitian nature of $\mathcal{L}$ (in our non-trivial inner product), and to spatial locality in our operator growth formalism (we can build an orthonormal operator basis by taking the tensor product of single-site operators).  We expect that no state with superluminal propagation exists; however, techniques which combine ours with those of \cite{kuwahara2021liebrobinson} may be required to definitively resolve this issue.

Looking forward, we anticipate our formalism will find wide applicability and generalizations.  First and foremost, our bound on information spreading in the 1d Bose-Hubbard model asymptotically agrees with previous numerics \cite{Light08,barmettler,Light14,Light18,Takasueaba9255,kennett} on the velocity of correlations both at low and high boson density.   Remarkably, this implies that there is no complicated, time-dependent protocol that can transmit information parametrically faster than simple time-evolution in the canonical 1d Bose-Hubbard model, perturbed away from the insulating state (which we would normally think of as having very slow dynamics)!  Our strong form of linear light cone, proven in 1d, also implies that (within the linear light cone) the Bose-Hubbard model is not much harder to simulate than an interacting spin model on a lattice;  this result may be somewhat surprising, as simulating Bose-Hubbard-like models has been conjectured to be a good experimental test of quantum supremacy \cite{Aaronson2011,Tran:2020xpc,maskara2020complexity,deshpande,Muraleedharan2018}. 

Since Corollary \ref{cor63} holds for arbitrary one-dimensional states with a finite number-density of bosons, our result rules out the possibility of using bosons to parametrically speed up quantum information transfer or signaling.   Regardless of the details of the microscopic time-dependent protocol, this result holds so long as the Hamiltonian only includes density-dependent interactions.  We anticipate this result can be extended to higher dimensions, but leave a proof to future work. 

There are many further scenarios where bounds on bosonic quantum information dynamics are highly desirable.  In trapped ion crystals \cite{Britton_2012} or cavity quantum electrodynamics \cite{thompson2020,Leroux_2010}, the Hamiltonian involves spins coupled to bosons.   Since these models do not typically conserve the number of bosons, our methods will need to be modified somewhat to remain applicable \cite{LRion}.  As both trapped ions and cavity-QED have been proposed as platforms for quantum computation or metrology, a fundamental speed limit on the time to implement a quantum gate (e.g.) is highly desirable.  Understanding bosonic dynamics in the presence of long-range hopping or interactions \cite{carleo} could also be important in generalizing our methods to these systems.  

Outside of quantum technologies, there are also interesting conjectures about fundamental speed limits on interacting phonons in metals \cite{Mousatov_2020}, which typically exhibit more complicated Hamiltonians than the Bose-Hubbard model.  Our new methods for studying bosonic dynamics may help to prove conjectured bounds in \cite{Mousatov_2020}.

Finally, there are many fundamental open questions about the nature of speed limits in \emph{finite temperature correlators}, which cannot be addressed using standard Lieb-Robinson techniques.  Recent results from gauge-gravity duality \cite{Maldacena:2015waa} have suggested universal temperature-dependent bounds on the emergent light cone that arises in finite temperature correlators. We expect that similar methods to those developed here -- namely, developing \emph{thermal} inner products on operator space \cite{lucasT} by replacing our $\rho=\mathrm{e}^{-\mu N}$ with $\rho = \mathrm{e}^{-\beta H}$ -- will aid in rigorously proving bounds on the thermal butterfly velocity that controls information spreading at finite temperature, which is believed exhibit universal dependencies on temperature \cite{blake,swingle}.  We have made early progress towards this question by affirmatively proving the universality of this conjectured temperature dependence in the dynamics of a single quantum particle \cite{chaofuture}; however, the extension to many-body systems remains an important open problem \cite{Han:2018bqy}.  If our methods can instead give strong bounds on butterfly velocities, they may further lead to a solution of longstanding challenges associated with whether -- and why -- the time scale of quantum dynamics at low temperature is always bounded by the ``Planckian time" $\hbar/k_{\mathrm{B}}T$ \cite{kss, Hartnoll_2014, lucasT}.   

While there remain many critical outstanding questions on the ``speed limits" on quantum dynamics, the quantum walk methods we have developed in this paper lead to a qualitatively new way of thinking about constraining quantum dynamics.  The methods introduced in this paper were particularly well-suited to tackling two important and common challenges that have arisen in the past: the unboundedness of the Hamiltonians of bosonic systems, and the desire to bound dynamics not in the entire Hilbert space, but only in an experimentally relevant (here, finite density) subspace.  We anticipate they could aid progress on the challenging mathematical physics problems highlighted above, in the near future.
\section*{Acknowledgements}
We thank Chi-Fang Chen and Abhinav Deshpande for useful feedback on a draft.  This work was supported by a Research Fellowship from the Alfred P. Sloan Foundation under Grant FG-2020-13795, and by the U.S. Air Force Office of Scientific Research under Grant FA9550-21-1-0195.

\bibliography{boson_OG}

\end{document}